\documentclass[a4paper,12pt]{article}
\usepackage[margin=1in]{geometry}
\usepackage[style=gb7714-2015]{biblatex}
\usepackage{hyperref,amsmath,amsthm,amssymb,amsthm,aligned-overset,tabularx,placeins,xcolor,listings,graphicx,subfigure}
\newtheorem{theorem}{Theorem}
\newtheorem{lemma}{Lemma}
\newtheorem{definition}{Definition}
\newtheorem{proposition}{Proposition}
\allowdisplaybreaks[4]
\providecommand{\keywords}[1]
{
  \small	
  \textbf{[Key words]:} #1
}
\lstdefinestyle{lfonts}{
  basicstyle   = \footnotesize\ttfamily,
  stringstyle  = \color{purple},
  keywordstyle = \color{blue!60!black}\bfseries,
  commentstyle = \color{olive}\scshape,
}
\lstdefinestyle{lnumbers}{
  numbers     = left,
  numberstyle = \tiny,
  numbersep   = 1em,
  firstnumber = 1,
  stepnumber  = 1,
}
\lstdefinestyle{llayout}{
  breaklines       = true,
  tabsize          = 2,
  columns          = flexible,
}
\lstdefinestyle{lgeometry}{
  xleftmargin      = 20pt,
  xrightmargin     = 0pt,
  frame            = tb,
  framesep         = \fboxsep,
  framexleftmargin = 20pt,
}
\lstdefinestyle{lgeneral}{
  style = lfonts,
  style = lnumbers,
  style = llayout,
  style = lgeometry,
}
\lstdefinestyle{python}{
    language = {Python},
    style    = lgeneral,
}
\title{Maximum Likelihood Estimates of Parameters in Generalized Gamma Distribution with SeLF Algorithm}
\author{Yufei Cai\footnote{The author's email address is \href{mailto:zhmgczh@gmail.com}{zhmgczh@gmail.com}.}}
\date{\today}
\begin{document}
\maketitle
\begin{abstract}
    This undergraduate thesis focuses on the problem of calculating maximum likelihood estimates (MLEs) of parameters in the generalized Gamma distribution using the SeLF algorithm. As an extension of the Gamma distribution, the generalized Gamma distribution can better fit real data and has been widely applied. The research begins by exploring the definition of the generalized Gamma distribution and its similarities and differences from the traditional Gamma distribution. Then, the SeLF and US algorithms are discussed in detail. The SeLF algorithm is a new algorithm based on the Minorization-Maximization (MM) algorithm, which can obtain the local optimal solution with few iterations, with the advantages of fast computation, high accuracy, and good convergence. The US algorithm is a method for finding the zeros of a function, which stands at a higher level than the SeLF algorithm and can improve the convergence speed and stability. This thesis proposes a method for calculating maximum likelihood estimates of the parameters in the generalized Gamma distribution using the SeLF and US algorithms, and presents the practical implementation of the algorithms, as well as simulations and data analysis to evaluate the performance of the proposed methods. The results demonstrate that the SeLF algorithm can achieve more stable and accurate estimates of the parameters in the generalized Gamma distribution more quickly, compared to traditional Newton's method, which can be useful in various applications. This thesis provides a comprehensive and in-depth exploration of the generalized Gamma distribution and the SeLF algorithm, and proposes a new method for calculating maximum likelihood estimates of parameters, contributing to the development of statistical methods for parameter estimation in complex models. It also provides a reference for further exploring the theory and application of the generalized Gamma distribution. The proposed method in this thesis has important practical significance and application value for solving practical problems.
\end{abstract}


\section{Generalized Gamma Distribution}

\subsection{Introduction}

The Gamma distribution is a commonly used probability distribution that has many applications in various fields, such as engineering, physics, and finance. It is a continuous probability distribution that is widely used to model waiting times, survival analysis, and lifetime data. However, there are situations where the Gamma distribution may not be an appropriate model due to its restrictive shape. In the paper \cite{stacy1962generalization}, \citeauthor{stacy1962generalization} has introduced a generalization of the Gamma distribution called the generalized Gamma distribution (GGD), which allows for more flexibility in the shape of the distribution. In this section, we discuss the properties and applications of this distribution and compare it to the (traditional) Gamma distribution.

\subsection{Definition of Generalized Gamma Distribution}

The generalized Gamma distribution is a three-parameter continuous probability distribution that has a flexible shape and can model a wider range of data than the Gamma distribution.

\begin{definition}[Generalized Gamma Distribution]
    The generalized Gamma distribution is a distribution with the general formula of the probability density function (PDF) given by:
    $$f(x;\alpha,\beta,\gamma)=\frac{\beta^\alpha\gamma}{\Gamma(\alpha/\gamma)}x^{\alpha-1}e^{-(\beta x)^\gamma},\ x>0$$
    where $\alpha,\beta>0$ are the shape parameters, $\gamma>0$ is the scale parameter and $\Gamma$ is the gamma function which is defined by $\Gamma(x)=\int_{0}^{\infty}t^{x-1}e^{-t}{\rm d}t$ for any $x>0$.
\end{definition}

The Gamma distribution is a special case of the generalized Gamma distribution when $\gamma=1$. When $\gamma>1$, the generalized Gamma distribution has a heavier tail than the Gamma distribution, while when $\gamma<1$, it has a lighter tail.

\subsection{Properties of Generalized Gamma Distribution}

The generalized Gamma distribution has several properties that make it a useful model for a wide range of data. Some of these properties are:

\begin{itemize}
    \item Flexibility: The generalized Gamma distribution can model a wide range of data shapes due to its flexibility. The shape of the distribution can be adjusted by varying the shape and scale parameters.

    \item Moment Generating Function: The moment generating function (MGF) of the generalized Gamma distribution is given by:
        $$
        M(t)=\sum_{n=0}^{\infty}\frac{(\frac{t}{\beta})^n}{n!}\cdot\frac{\Gamma(\frac{\alpha+n}{\gamma})}{\Gamma(\frac{\alpha}{\gamma})}
        $$
        where $\Gamma$ denotes the Gamma function. This can be used to derive the moments of the distribution.

    \item Cumulative Distribution Function: The cumulative distribution function (CDF) of the generalized Gamma distribution does not have a closed-form solution. However, it can be approximated using numerical methods.
\end{itemize}

\subsection{Applications of Generalized Gamma Distribution}

The generalized Gamma distribution has several applications in various fields. Some of these applications are:

\begin{itemize}
    \item Lifetime Data: The generalized Gamma distribution can be used to model lifetime data, such as the time to failure of a component in a machine. For a more detailed discussion, please refer to \citeauthor{shanker2016modeling}'s paper \cite{shanker2016modeling}.

    \item Wind Speed: The generalized Gamma distribution can be used to model wind speed data, which is important in the design of wind turbines. For a more detailed discussion, please refer to the paper \cite{mert2015statistical} by \citeauthor{mert2015statistical}.

    \item Finance: The generalized Gamma distribution can be used to model the distribution of stock returns, which can help in the analysis of financial data. For a more detailed discussion, please refer to \citeauthor{gomes2006four}'s paper \cite{gomes2006four}.
\end{itemize}

\subsection{Comparison with Gamma Distribution}
The gamma distribution is a well-known distribution used in many fields, including statistics, physics, engineering, and finance. The generalized Gamma distribution (GGD) is a generalization of the gamma distribution that offers more flexibility in modeling real-world data. In this subsection of the thesis, we compare and contrast the two distributions and highlight the advantages of the GGD.

\begin{definition}[Gamma Distribution]
    The gamma distribution is a family of continuous probability distributions with two parameters, including a shape parameter $\alpha$ and a scale parameter $\beta$. Its probability density function (PDF) is given by:
    $$f(x;\alpha,\beta)=\frac{\beta^{\alpha}}{\Gamma(\alpha)}x^{\alpha-1}e^{-\beta x},\ x>0$$
    where $\Gamma(\alpha)$ is the gamma function.
\end{definition}

The gamma distribution has many applications, including in reliability analysis, queuing theory, and finance. However, it is limited in its ability to model data with different shapes and tail behavior. In particular, the gamma distribution has a monotonically decreasing hazard function, which means that the failure rate decreases over time.

The generalized Gamma distribution is a more flexible distribution than the gamma distribution because it allows for different tail behavior and can model data with more complex shapes. The parameter $\gamma$ controls the shape of the distribution, and when $\gamma=1$, the generalized Gamma distribution reduces to the gamma distribution. The generalized Gamma distribution has been used in many applications, including in modeling rainfall, wind speed, and response times in neuroscience.

The gamma distribution is a special case of the generalized Gamma distribution, and therefore, it is less flexible than the generalized Gamma distribution. The generalized Gamma distribution can model data with more complex shapes and tail behavior, making it more suitable for real-world applications. In addition, the generalized Gamma distribution has a more interpretable shape parameter $\gamma$, which allows for easier interpretation of the distribution's properties. However, the gamma distribution is simpler to work with and has a closed-form solution for many statistical problems, making it more computationally efficient.

In summary, the generalized Gamma distribution is a more flexible distribution than the gamma distribution and offers more options for modeling real-world data. However, the gamma distribution is simpler to work with and has closed-form solutions for many statistical problems. The choice between the two distributions depends on the specific application and the goals of the analysis.

\section{Second–derivative Lower–bound Function Algorithm}

\subsection{Introduction}

The paper \cite{tian2025second} by \citeauthor{tian2025second} has proposed the \textit{second-derivative lower-bound function} (SeLF) algorithm, which is a novel maximization method based on the framework of \textit{minorization-maximization} (MM) algorithms. The SeLF algorithm is a useful tool for calculating the maximum likelihood estimate (MLE) $\hat{\theta}$ of some parameter in a one-dimensional log-likelihood function $l(\theta)$ (or maximize a more generalized target function). In order to introduce the SeLF algorithm more clearly, we have to spend some space introducing MM algorithms first, because the SeLF algorithm is based on the fundamental principles of MM algorithms.

\subsection{Minorization-Maximization Algorithms}

Minorization-Maximization (MM) algorithms are a family of iterative optimization algorithms that are commonly used in machine learning and statistical inference. These algorithms are particularly useful for solving optimization problems that involve log-concave likelihood functions or other functions that are difficult to optimize directly. In this subsection, we will provide a formal definition of MM algorithms and discuss some of their key properties and applications. \citeauthor{hunter2004tutorial} compiled the general principles of MM algorithms in their paper \cite{hunter2004tutorial}.

\subsubsection{Definition}

The basic idea behind MM algorithms is to construct a sequence of simpler functions that ``minorize'' the objective function of interest, and then optimize each of these simpler functions in turn. At each iteration, the algorithm chooses a minorizing function that is a lower bound on the objective function, and then optimizes this minorizing function. This optimization step is known as the ``maximization'' step, since we are maximizing a lower bound on the objective function.

More formally, suppose we wish to maximize a function $l(\theta)$ over some parameter space $\Theta$. We define a surrogate function $Q(\theta|\theta^{(t)})$ as a function that depends on the current estimate $\theta^{(t)}$, and that satisfies the following two conditions:

\begin{enumerate}
\item $Q(\theta|\theta^{(t)})\leq l(\theta)$ for all $\theta \in \Theta$
\item $Q(\theta^{(t)}|\theta^{(t)})=l(\theta^{(t)})$
\end{enumerate}

Condition (1) ensures that $Q(\theta|\theta^{(t)})$ is a lower bound on $l(\theta)$, while condition (2) ensures that the bound is tight at the current estimate $\theta^{(t)}$. Note that there may be many functions that satisfy these conditions; the choice of $Q(\theta|\theta^{(t)})$ will depend on the specific problem being solved. The process of finding the $Q$ function is usually referred to as the minorization step.

Given a minorizing (or surrogate) function $Q(\theta|\theta^{(t)})$, we define the next estimate in the sequence as:

\begin{align*}
\theta^{(t+1)}=\arg\max_{\theta\in\Theta}Q(\theta|\theta^{(t)}).
\end{align*}

That is, we choose the value of $\theta$ that maximizes the lower bound $Q(\theta|\theta^{(t)})$. This step is known as the maximization step.

The iteration continues until convergence, typically when the change in the estimate between iterations falls below a predefined tolerance level.

\subsubsection{Properties}

MM algorithms have several useful properties that make them well-suited for a variety of optimization problems.

\paragraph{Convergence}

One of the most important properties of an optimization algorithm is convergence. MM algorithms have been shown to converge under certain conditions. Specifically, if $l(\theta)$ is a concave function and the surrogate function $Q(\theta|\theta^{(t)})$ satisfies the following conditions:

\begin{enumerate}
\item $Q(\theta^{(t)}|\theta^{(t)})=l(\theta^{(t)})$
\item $Q(\theta|\theta^{(t)})\le l(\theta)$ for all $\theta\in\Theta$
\item $\theta^{(t+1)}=\arg\max_\theta Q(\theta|\theta^{(t)})$
\end{enumerate}

then the MM algorithm is guaranteed to converge to a stationary point of $l(\theta)$.

\paragraph{Monotonicity}

Another important property of optimization algorithms is monotonicity, which means that the objective function value improves or stays the same at each iteration. MM algorithms are not necessarily monotonic since there is another type of minorization-maximization algorithm (usually referred to as generalized minorization-maximization algorithm) which only requires the $Q$ function to be smaller than the $l$ function at any point $\theta\in\Theta$ yet without requiring the $Q$ function and the $l$ function to be equal at point $\theta^{(t)}$, but under certain conditions, they can be monotonic.

Specifically, if the surrogate function $Q(\theta|\theta^{(t)})$ satisfies the following condition:

\begin{itemize}
\item $Q(\theta^{(t)}|\theta^{(t)})=l(\theta^{(t)})$
\end{itemize}

then the MM algorithm is monotonic, meaning that $l(\theta^{(t+1)})\geq l(\theta^{(t)})$ for all $t$.

\paragraph{Stability}

Stability is another important property of optimization algorithms. A stable algorithm is one that does not produce wildly varying solutions for small changes in the input or parameters. MM algorithms are generally stable, but they can exhibit instability if the surrogate function is not well-designed.

To ensure stability, the surrogate function $Q(\theta|\theta^{(t)})$ should satisfy the following condition:

\begin{itemize}
\item $|Q(\theta|\theta^{(t)})-Q(\theta|\theta^{(t')})|\leq C|\theta^{(t)}-\theta^{(t')}|$
\end{itemize}

where $C$ is a constant that depends on the choice of surrogate function and $\theta^{(t')}$ is a nearby point to $\theta^{(t)}$. This condition ensures that small changes in the input or parameters result in small changes in the surrogate function.

In summary, minorization-maximization (MM) algorithms are a class of iterative optimization methods that are useful for maximizing difficult objective functions. They are guaranteed to converge to a stationary point under certain conditions and can be monotonic under additional conditions. Furthermore, they are generally stable, but stability depends on the choice of surrogate function.

\subsection{SeLF algorithm}

\subsubsection{Definition}

The paper \cite{tian2025second} proposes a novel minorization-maximization (MM) method known as the \textit{second-derivative lower-bound function} (SeLF) algorithm, which is a general approach for iteratively computing the maximum likelihood estimate (MLE) $\hat{\theta}$ of the parameter $\theta$ in a one-dimensional target function of interest, typically the marginal likelihood function $l(\theta)$.

The SeLF algorithm consists of two steps per iteration: a second-derivative lower-bound function step (SeLF-step), which is a special type of minorization step, and a maximization step (M-step).

Let $l(\theta)$ be a one-dimensional target function of interest, such as the log-likelihood function of the parameter $\theta$, and assume that the first-order derivative $l'(\theta)$ and the second-order derivative $l''(\theta)$ exist. We seek to obtain the MLE of $\theta$ such that $\hat{\theta}=\arg\max_{\theta\in\Theta}l(\theta)$, where $\Theta\subseteq\mathbb{R}\triangleq(-\infty,\infty)$ is the parameter space.

Suppose that there exists a function $b(\theta)$ where $\theta\in\Theta$ such that $l''(\theta)\ge b(\theta)$ for all $\theta\in\Theta$, where $b(\theta)$ is called the \textit{second-derivative lower-bound function} (SeLF) for $l(\theta)$. To construct a minorizing function based on the SeLF $b(\theta)$, we first define
$$
B(\theta|\theta^{(t)})\triangleq\int_{\theta^{(t)}}^{\theta}b(z){\rm d}z
$$
and then we introduce the following auxiliary function
$$
h(\theta|\theta^{(t)})\triangleq l(\theta)-\int_{\theta^{(t)}}^{\theta}B(z|\theta^{(t)}){\rm d}z.
$$
As a result,
\begin{align*}
    h'(\theta|\theta^{(t)})=&l'(\theta)-B(\theta|\theta^{(t)}),\\
    h''(\theta|\theta^{(t)})=&l''(\theta)-b(\theta)\ge0,\\
    h(\theta^{(t)}|\theta^{(t)})=&l(\theta^{(t)}),\\
    h'(\theta^{(t)}|\theta^{(t)})=&l'(\theta^{(t)})-B(\theta^{(t)}|\theta^{(t)})=l'(\theta^{(t)}).
\end{align*}
Since $h''(\theta|\theta^{(t)})\ge0$, indicating that $h(\theta)$ is convex, by the Taylor formula from analysis, we can see that for some specific $\theta^*$ between $\theta$ and $\theta^{(t)}$
\begin{align*}
    h(\theta|\theta^{(t)})&=h(\theta^{(t)}|\theta^{(t)})+h'(\theta^{(t)}|\theta^{(t)})(\theta-\theta^{(t)})+\frac{1}{2}h''(\theta^*)(\theta-\theta^*)^2\\
    &\ge h(\theta^{(t)}|\theta^{(t)})+h'(\theta^{(t)}|\theta^{(t)})(\theta-\theta^{(t)})
\end{align*}
where the equality holds iff $\theta=\theta^{(t)}$. Therefore, we obtain
\begin{align*}
    l(\theta)&=h(\theta|\theta^{(t)})+\int_{\theta^{(t)}}^{\theta}B(z|\theta^{(t)}){\rm d}z\\
    &\ge h(\theta^{(t)}|\theta^{(t)})+h'(\theta^{(t)}|\theta^{(t)})(\theta-\theta^{(t)})+\int_{\theta^{(t)}}^{\theta}B(z|\theta^{(t)}){\rm d}z\\
    &=l(\theta^{(t)})+l'(\theta^{(t)})(\theta-\theta^{(t)})+\int_{\theta^{(t)}}^{\theta}B(z|\theta^{(t)}){\rm d}z\\
    &=l'(\theta^{(t)})\theta+\int_{\theta^{(t)}}^{\theta}B(z|\theta^{(t)}){\rm d}z+c^{(t)}
\end{align*}
which minorizes $l(\theta)$ at $\theta=\theta^{(t)}$ where $c^{(t)}=l(\theta^{(t)})-l'(\theta^{(t)})\theta^{(t)}$ is a constant free from $\theta$. So we can define $Q(\theta|\theta^{(t)})\triangleq l'(\theta^{(t)})\theta+\int_{\theta^{(t)}}^{\theta}B(z|\theta^{(t)}){\rm d}z+c^{(t)}$ as the surrogate function of the MM algorithm.

According to the M-step of minorization-maximization algorithm, letting ${\rm d}Q(\theta|\theta^{(t)})/{\rm d}\theta=0$, the maximizer of the $Q(\theta|\theta^{(t)})$ function is
\begin{align*}
    \theta^{(t+1)}&=\arg\max_{\theta\in\Theta}Q(\theta|\theta^{(t)})\\
    &=\arg\max_{\theta\in\Theta}\left[l'(\theta^{(t)})\theta+\int_{\theta^{(t)}}^{\theta}B(z|\theta^{(t)}){\rm d}z\right]\\
    &={\rm sol}\left\{\theta\in\Theta:l'(\theta^{(t)})+B(\theta|\theta^{(t)})=0\right\}\\
    &={\rm sol}\left\{\theta\in\Theta:l'(\theta^{(t)})+\int_{\theta^{(t)}}^{\theta}b(z){\rm d}z=0\right\}
\end{align*}
where $\theta^{(t)}\in\Theta$.

\subsubsection{Properties}

The SeLF algorithm has several useful properties. It not only inherits all the excellent characteristics of the MM algorithm but also possesses more specific convergence and stability.

Several important properties are described by the paper \cite{tian2025second} as follows.

\paragraph{Convergence}

\begin{definition}[Weakly Stable Convergence\cite{tian2025second}]
    A sequence $\{\theta^{(t)}\}_{t=0}^\infty$ is said to \textit{weakly stably converge} to a fixed point $\theta^*$ if the distance between $\theta^{(t)}$ and $\theta^*$ decreases at each iteration, i.e., $|\theta^{(t+1)}-\theta^*|<|\theta^{(t)}-\theta^*|\ \forall\ t\in\mathbb{N}$.
\end{definition}
    
\begin{definition}[Strongly Stable Convergence\cite{tian2025second}]
    A sequence $\{\theta^{(t)}\}_{t=0}^\infty$ is said to \textit{strongly stably converge} to a fixed point $\theta^*$ if it approaches $\theta^*$ in a strictly monotonically increasing or decreasing manner, i.e., $\theta^{(0)}<\theta^{(1)}<\cdots<\theta^{(t)}<\cdots\le\theta^*$ or $\theta^*\le\cdots<\theta^{(t)}<\cdots<\theta^{(1)}<\theta^{(0)}$.
\end{definition}
    
An algorithm is said to have \textit{weakly (or strongly) stable convergence} if the sequence $\{\theta^{(t)}\}_{t=0}^\infty$ generated by the algorithm weakly (or strongly) stably converges to a fixed point $\theta^*$. It is worth noting that strong stable convergence implies weak stable convergence.

The paper \cite{tian2025second} established the following theorem:

\begin{theorem}[Strongly Stable Convergence\cite{tian2025second}]\label{strongly_stable_convergence}
Let $\hat{\theta}$ denote the MLE of $\theta$, and let $b(\theta)$ be a SeLF for $l(\theta)$ in $\Theta$. Then, for any initial value $\theta^{(0)}\in\Theta$, the sequence $\{\theta^{(t+1)}\}_{t=0}^\infty$ generated by the SeLF algorithm strongly stably converges to $\hat{\theta}$.
\end{theorem}

\paragraph{Stability}

The proof of Theorem \ref{strongly_stable_convergence} in the paper \cite{tian2025second} also guarantees the stability property.

In fact, Theorem \ref{strongly_stable_convergence} establishes that the SeLF algorithm possesses two desirable properties: (\romannumeral1) strongly stable convergence to the MLE $\hat{\theta}$; (\romannumeral2) independence from any initial values in $\Theta$, which sets it apart from the Newton algorithm.

\section{Upper-crossing/Solution Algorithm}

\subsection{Introduction}

In \citeauthor{li2022upper}'s recent work \cite{li2022upper}, \citeauthor{li2022upper} introduces a new and versatile root-finding method, named the \textit{Upper-crossing/Solution} (US) algorithm, which belongs to the class of non-bracketing (or open domain) methods. The US algorithm offers a general approach to iteratively determine the unique root $\theta^*$ of a non-linear equation $g(\theta)=0$. Each iteration of the US algorithm consists of two steps: an upper-crossing step (U-step) and a solution step (S-step). During the U-step, the algorithm seeks an upper-crossing function or a $U$-function, denoted by $U(\theta|\theta^{(t)})$, which takes a particular form depending on the $t$-th iteration $\theta^{(t)}$ of $\theta^*$, and is derived using a novel concept called the changing direction inequality. In the subsequent S-step, the US algorithm solves the simple $U$-equation, $U(\theta|\theta^{(t)})=0$, to obtain the explicit solution $\theta^{(t+1)}$.

\subsection{Definition}

\begin{definition}[Changing Direction (CD) Inequalities\cite{li2022upper}]
For two functions $h_1(x)$ and $h_2(x)$ with the same domain $\mathbb{X}$,
$$
h_1(x)\overset{{\rm sgn}(a)}{\ge}h_2(x)\ \Longleftrightarrow\ \begin{cases}
    h_1(x)\ge h_2(x),&{\rm if}\ a>0,\\
    h_1(x)=h_2(x),&{\rm if}\ a=0,\\
    h_1(x)\le h_2(x),&{\rm if}\ a<0,
\end{cases}
$$
and
$$
h_1(x)\overset{{\rm sgn}(a)}{\le}h_2(x)\ \Longleftrightarrow\ \begin{cases}
    h_1(x)\le h_2(x),&{\rm if}\ a>0,\\
    h_1(x)=h_2(x),&{\rm if}\ a=0,\\
    h_1(x)\ge h_2(x),&{\rm if}\ a<0.
\end{cases}
$$
\end{definition}

Suppose that solving directly for the unique root $\theta^*$ of the non-linear equation $g(\theta)=0,\ \theta\in\Theta\subset\mathbb{R}$ is extremely challenging. Without loss of generality, we assume that
$$
g(\theta)>0,\ \forall\ \theta<\theta^*\ {\rm and}\ g(\theta)<0,\ \forall\ \theta>\theta^*.
$$
Alternatively, we can simplify the assumption with CD inequalities as $g(\theta)\overset{{\rm sgn}(\theta-\theta^*)}{\le}0$ for all $\theta\in\Theta$. If $g(\theta)<0$ when $\theta<\theta^*$, then we can multiply both sides of the nonlinear equation $g(\theta)=0,\ \theta\in\Theta\subset\mathbb{R}$ by $-1$ to obtain a new equation $-g(\theta)=0$, which satisfies the above assumption.

The idea behind the US algorithm is as follows: Let $\theta^{(t)}$ denote the $t$-th iteration of $\theta^*$ and $U(\theta|\theta^{(t)})$ denote a real-valued function of $\theta$ whose form depends on $\theta^{(t)}$. The function $U(\theta|\theta^{(t)})$ is called a $U$-function for $g(\theta)$ at $\theta=\theta^{(t)}$ if
$$
g(\theta)\overset{{\rm sgn}(\theta-\theta^{(t)})}{\ge}U(\theta|\theta^{(t)}),\ \forall\ \theta\in\Theta
$$
where $\theta^{(t)}\in\Theta$.

The conditions indicate that the $U(\theta|\theta^{(t)})$ function dominates the objective function $g(\theta)$ at $\theta=\theta^{(t)}$ for all $\theta\le\theta^{(t)}$ and the $U(\theta|\theta^{(t)})$ function underestimates $g(\theta)$ at $\theta=\theta^{(t)}$ for all $\theta\ge\theta^{(t)}$. Using the definition of the $U$-function, we conclude that $g(\theta)$ is a $U$-function for $0$ at $\theta=\theta^*$.

Some basic properties of $U$-functions are given by the paper \cite{li2022upper} as follows.

\begin{proposition}[Basic Properties of $U$-functions\cite{li2022upper}]
    \quad
    \begin{enumerate}
        \item[(a)] Let $g(\theta)$ be defined in $\Theta$, then $g(\theta)$ is a $U$-function for itself at any interior of $\Theta$;
        \item[(b)] If $U_1(\theta|\theta^{(t)})$ and $U_2(\theta|\theta^{(t)})$ are two $U$-functions for $g(\theta)$ at $\theta=\theta^{(t)}$, then $[U_1(\theta|\theta^{(t)})+U_2(\theta|\theta^{(t)})]/2$ is also a $U$-function for $g(\theta)$ at $\theta=\theta^{(t)}$;
        \item[(c)] If $g_1(\theta)$ and $g_2(\theta)$ are two functions defined on $\Theta$ and $U_1(\theta|\theta^{(t)})$ is a $U$-function for $g_1(\theta)$ at $\theta=\theta^{(t)}$, then $U_1(\theta|\theta^{(t)})+g_2(\theta)$ is a $U$-function for $g_1(\theta)+g_2(\theta)$ at $\theta=\theta^{(t)}$.
    \end{enumerate}
\end{proposition}

If a $U$-function can be found, one needs only solve the simple $U$-equation: $U(\theta|\theta^{(t)})=0$ to obtain its solution $\theta^{(t+1)}$. This is denoted as
$$
\theta^{(t+1)}={\rm sol}\{\theta\in\Theta:U(\theta|\theta^{(t)})=0\}
$$
where $\theta^{(t)}\in\Theta$.

Therefore, the US algorithm is a general principle for iteratively seeking the root $\theta^*$ with strongly stable convergence. Each iteration consists of an upper-crossing step (\textbf{U-step}) and a solution step (\textbf{S-step}), where
\begin{table*}[h!]
    \begin{center}
        \begin{tabularx}{\textwidth}{lXXX}
            \textbf{U-step:}&Find a $U$-function $U(\theta|\theta^{(t)})$ satisfying $g(\theta)\overset{{\rm sgn}(\theta-\theta^{(t)})}{\ge}U(\theta|\theta^{(t)}),\ \forall\ \theta\in\Theta$;\\
            \textbf{S-step:}&Solve the $U$-equation to obtain its root $\theta^{(t+1)}$ as showed by $\theta^{(t+1)}={\rm sol}\{\theta\in\Theta:U(\theta|\theta^{(t)})=0\}$.
        \end{tabularx}
    \end{center}
\end{table*}
\FloatBarrier

Typically, the $U$-equation $U(\theta|\theta^{(t)})=0$ has an analytic solution, such as being a polynomial with degree less than 5 (or even a linear equation). Hence, the US method involves iteratively solving a sequence of tractable surrogate equations $U(\theta|\theta^{(t)})=0$, instead of the intractable original nonlinear equation $g(\theta)=0$.

\subsection{Properties}

The US algorithm, as a new algorithm, is not a special case of the MM algorithm. On the contrary, its concept is more open than the MM algorithm. However, this does not prevent it from having very good convergence and stability. The paper \cite{li2022upper} proves several useful properties, which help us see the power of the US algorithm.

\paragraph{Convergence}

\begin{theorem}[Strongly Stable Convergence\cite{li2022upper}]\label{us_strongly_stable_convergence}
    Let $\theta^*$ be the unique root of the equation $g(\theta)=0,\ \theta\in\Theta\subset\mathbb{R}$, and let $U(\theta|\theta^{(t)})$ be a $U$-function for $g(\theta)$ at $\theta=\theta^{(t)}$ in $\Theta$. Then, for any initial value $\theta^{(0)}$ in $\Theta$, the sequence $\{\theta^{(t+1)}\}_{t=0}^\infty$ defined by $\theta^{(t+1)}={\rm sol}\{\theta\in\Theta:U(\theta|\theta^{(t)})=0\}$ strongly and stably converges to $\theta^*$.
\end{theorem}

\paragraph{Stablility}

The proof of Theorem \ref{us_strongly_stable_convergence} in the paper \cite{li2022upper} also guarantees the stability property.

Moreover, Theorem \ref{us_strongly_stable_convergence} shows that the US algorithm has two desirable properties: (\romannumeral1) strong stable convergence to the root $\theta^*$; (\romannumeral2) independence from any initial values in $\Theta$, which sets it apart from the Newton algorithm.

\subsection{Relationship with SeLF Algorithm}

In fact, we can observe that the SeLF algorithm can be viewed as a special case of the US algorithm if we consider that finding the maximum likelihood estimate is equivalent to finding the zero of the log-likelihood function or the root of the equation $l'(\theta)=0$. The paper \cite{li2022upper} presents four methods for constructing $U$-functions: the FLB function method, the SLUB constants method, the TLB method, and the fixed-block method. To prove that the SeLF algorithm is a special case of the US algorithm, we focus only on the FLB function method in this discussion.

\begin{theorem}[The FLB Function Method\cite{li2022upper}]\label{flb_function_method}
    Suppose the first-order derivative $g'(\cdot)$ exists and it is bounded by some function $b(\cdot)$; i.e.,
    $$
    g'(\theta)\ge b(\theta),\ \forall\ \theta\in\Theta,
    $$
    where $b(\theta)$ is called the first-derivative lower bound (FLB) function. Then
    $$
    U(\theta|\theta^{(t)})\triangleq g(\theta^{(t)})+\int_{\theta^{(t)}}^{\theta}b(z){\rm d}z,\ \forall\ \theta,\theta^{(t)}\in\Theta,
    $$
    is a $U$-function for $g(\theta)$ at $\theta=\theta^{(t)}$.
\end{theorem}
\begin{proof}\cite{li2022upper}
    \begin{align*}
        g(\theta)-U(\theta|\theta^{(t)})&=[g(\theta)-g(\theta^{(t)})]-\int_{\theta^{(t)}}^{\theta}b(z){\rm d}z=\int_{\theta^{(t)}}^{\theta}g'(z){\rm d}z-\int_{\theta^{(t)}}^{\theta}b(z){\rm d}z\\
        &=\int_{\theta^{(t)}}^{\theta}[g'(z)-b(z)]{\rm d}z\begin{cases}
            \le0,&{\rm if}\ \theta<\theta^{(t)},\\
            =0,&{\rm if}\ \theta=\theta^{(t)},\\
            \ge0,&{\rm if}\ \theta>\theta^{(t)}
        \end{cases}\\
        \overset{{\rm sgn}(\theta-\theta^{(t)})}&{\ge}0,\ \forall\ \theta,\theta^{(t)}\in\Theta.
    \end{align*}
\end{proof}

\begin{proposition}\label{equivalence}
    The SeLF algorithm can be viewed as a special case of the US algorithm if we consider that finding the maximum likelihood estimate is equivalent to finding the zero of the log-likelihood function or the root of the equation $l'(\theta)=0$.
\end{proposition}
\begin{proof}
    The SeLF algorithm aims to find the maximum of the log-likelihood function $l(\theta)$ or, equivalently, the zero of the first derivative $l'(\theta)$. Let $g(\theta) = l'(\theta)$ for $\theta \in \Theta$, where $\Theta$ is the parameter space. The SeLF algorithm uses the FLB function method to obtain a lower bound function $b(\theta)$ of $l''(\theta)$ such that $l''(\theta) \geq b(\theta)$ for all $\theta\in\Theta$, which implies $g'(\theta) \geq b(\theta)$ for all $\theta\in\Theta$.
    
    By Theorem \ref{flb_function_method} in the paper \cite{li2022upper}, $U(\theta|\theta^{(t)})\triangleq g(\theta^{(t)})+\int_{\theta^{(t)}}^{\theta}b(z){\rm d}z=l'(\theta^{(t)})+\int_{\theta^{(t)}}^{\theta}b(z){\rm d}z$ is a $U$-function of $g(\theta)$ in $\Theta$. Therefore, each iteration of the SeLF algorithm is $\theta^{(t+1)}={\rm sol}\left\{\theta\in\Theta:l'(\theta^{(t)})+\int_{\theta^{(t)}}^{\theta}b(z){\rm d}z=0\right\}$.
    
    Similarly, the US algorithm aims to find the root of the function $g(\theta)=l'(\theta)$ with a $U$-function $U(\theta|\theta^{(t)})$. Each iteration of the US algorithm is $\theta^{(t+1)}={\rm sol}
    \{\theta\in\Theta:U(\theta|\theta^{(t)})=0\}={\rm sol}\left\{\theta\in\Theta:g(\theta^{(t)})+\int_{\theta^{(t)}}^{\theta}b(z){\rm d}z=0\right\}$.
    
    Since $l'(\theta^{(t)})=g(\theta^{(t)})$, we have that each iteration of the SeLF algorithm and the US algorithm is the same. Therefore, the SeLF algorithm can be viewed as a special case of the US algorithm if we consider that finding the maximum likelihood estimate is equivalent to finding the zero of the log-likelihood function or the root of the equation $l'(\theta)=0$.
\end{proof}

\section{MLE in Generalized Gamma Distribution}

In this section, we will utilize the SeLF algorithm to estimate the parameters of the generalized Gamma distribution. As we have established some kind of equivalence between the US algorithm and the SeLF algorithm through Proposition \ref{equivalence}, we can directly apply the simpler SeLF algorithm to achieve our goal without the need to design a separate US algorithm. Essentially, when we use the SeLF algorithm to estimate the parameters of the generalized Gamma distribution, we are utilizing the US algorithm.

Suppose there are $n$ random variables $\{X_i\}_{i=1}^n$ which are mutually independent and follow the same generalized Gamma distribution with the probability density function defined as $f(x;\alpha,\beta,\gamma)=\frac{\beta^\alpha\gamma}{\Gamma(\alpha/\gamma)}x^{\alpha-1}e^{-(\beta x)^\gamma}$ where $\alpha,\beta,\gamma>0$. Let $\{x_i\}_{i=1}^n$ be a realization of $\{X_i\}_{i=1}^n$. The likelihood function is defined as
$$
L(\alpha,\beta,\gamma)=\prod_{i=1}^nf(x_i;\alpha,\beta,\gamma).
$$
The log-likelihood function is then defined as
\begin{align*}
    l(\alpha,\beta,\gamma)=&\log\prod_{i=1}^nf(x_i;\alpha,\beta,\gamma)\\
    =&\sum_{i=1}^n\log f(x_i;\alpha,\beta,\gamma)\\
    =&\sum_{i=1}^n\log\left[\frac{\beta^\alpha\gamma}{\Gamma(\alpha/\gamma)}x_i^{\alpha-1}e^{-(\beta x_i)^\gamma}\right]\\
    =&\sum_{i=1}^n[\alpha\log\beta+\log\gamma-\log\Gamma(\alpha/\gamma)+(\alpha-1)\log x_i-(\beta x_i)^\gamma]\\
    =&n[\alpha\log\beta+\log\gamma-\log\Gamma(\alpha/\gamma)]+(\alpha-1)\sum_{i=1}^n\log x_i-\beta^\gamma\sum_{i=1}^nx_i^\gamma
\end{align*}

The first- and second-order partial derivatives of $l(\alpha,\beta,\gamma)$ with respect to $\alpha$ are:
\begin{align*}
    \frac{\partial l(\alpha,\beta,\gamma)}{\partial\alpha}=&n\left[\log\beta-\frac{1}{\gamma}\frac{\Gamma'(\alpha/\gamma)}{\Gamma(\alpha/\gamma)}\right]+\sum_{i=1}^n\log x_i\\
    =&n\left\{\log\beta-\frac{1}{\gamma}\left[-\gamma_0+\sum_{m=0}^\infty\left(\frac{1}{m+1}-\frac{1}{m+\alpha/\gamma}\right)\right]\right\}+\sum_{i=1}^n\log x_i\\
    \frac{\partial^2l(\alpha,\beta,\gamma)}{\partial\alpha^2}=&-\frac{n}{\gamma^2}\sum_{m=0}^\infty\frac{1}{(m+\alpha/\gamma)^2}
\end{align*}
where $\gamma_0=\lim_{n\to\infty}\left(\sum_{m=1}^nm^{-1}-\log n\right)\approx0.5772$ is the Euler-Mascheroni constant.

Since $\alpha,\gamma>0$, $\alpha/\gamma>0$, so
\begin{align*}
    \frac{\partial^2l(\alpha,\beta,\gamma)}{\partial\alpha^2}=&-\frac{n}{\gamma^2}\sum_{m=0}^\infty\frac{1}{(m+\alpha/\gamma)^2}\\
    =&-\frac{n}{\gamma^2}\left(\frac{1}{(\alpha/\gamma)^2}+\sum_{m=1}^\infty\frac{1}{(m+\alpha/\gamma)^2}\right)\\
    \ge&-\frac{n}{\gamma^2}\left(\frac{1}{(\alpha/\gamma)^2}+\sum_{m=1}^\infty\frac{1}{m^2}\right)\\
    =&-\frac{n}{\alpha^2}-\frac{n}{\gamma^2}\sum_{m=1}^{\infty}\frac{1}{m^2}\\
    =&-\frac{n}{\alpha^2}-\frac{\pi^2}{6}\cdot\frac{n}{\gamma^2}\\
    \triangleq&b_\alpha(\alpha,\beta,\gamma).
\end{align*}

Then we can construct a $U$-function of $\frac{\partial l(\alpha,\beta,\gamma)}{\partial\alpha}$ by $b_\alpha(\alpha,\beta,\gamma)$ as follows:
\begin{align*}
    &U_\alpha(\alpha,\beta,\gamma|\alpha^{(t)},\beta^{(t)},\gamma^{(t)})\\
    =&\frac{\partial l(\alpha,\beta,\gamma)}{\partial\alpha}(\alpha^{(t)},\beta^{(t)},\gamma^{(t)})+\int_{\alpha^{(t)}}^\alpha b_\alpha(z,\beta^{(t)},\gamma^{(t)}){\rm d}z\\
    =&n\left\{\log\beta^{(t)}-\frac{1}{\gamma^{(t)}}\left[-\gamma_0+\sum_{m=0}^\infty\left(\frac{1}{m+1}-\frac{1}{m+\alpha^{(t)}/\gamma^{(t)}}\right)\right]\right\}+\sum_{i=1}^n\log x_i\\
    &+\left[\frac{n}{z}-\frac{\pi^2}{6}\cdot\frac{n}{(\gamma^{(t)})^2}z\right]_{\alpha^{(t)}}^\alpha\\
    =&n\left\{\log\beta^{(t)}-\frac{1}{\gamma^{(t)}}\left[-\gamma_0+\sum_{m=0}^\infty\left(\frac{1}{m+1}-\frac{1}{m+\alpha^{(t)}/\gamma^{(t)}}\right)\right]\right\}+\sum_{i=1}^n\log x_i\\
    &-\left[\frac{n}{\alpha^{(t)}}-\frac{\pi^2}{6}\cdot\frac{n\alpha^{(t)}}{(\gamma^{(t)})^2}\right]+\left[\frac{n}{\alpha}-\frac{\pi^2}{6}\cdot\frac{n\alpha}{(\gamma^{(t)})^2}\right]
\end{align*}

Letting $U_\alpha(\alpha,\beta,\gamma|\alpha^{(t)},\beta^{(t)},\gamma^{(t)})=0$, we get
$$
\alpha^{(t+1)}={\rm sol}\left\{\alpha>0:a^{(t)}\alpha^2+b^{(t)}\alpha+c^{(t)}=0\right\}
$$
where $a^{(t)}=-\frac{\pi^2}{6}\cdot\frac{n}{(\gamma^{(t)})^2}$, $b^{(t)}=n\left\{\log\beta^{(t)}-\frac{1}{\gamma^{(t)}}\left[-\gamma_0+\sum_{m=0}^\infty\left(\frac{1}{m+1}-\frac{1}{m+\alpha^{(t)}/\gamma^{(t)}}\right)\right]\right\}+\sum_{i=1}^n\log x_i-\frac{n}{\alpha^{(t)}}+\frac{\pi^2}{6}\cdot\frac{n\alpha^{(t)}}{(\gamma^{(t)})^2}$ and $c^{(t)}=n$.

The first-order partial derivative of $l(\alpha,\beta,\gamma)$ with respect to $\beta$ is
$$
\frac{\partial l(\alpha,\beta,\gamma)}{\partial\beta}=\frac{n\alpha}{\beta}-\gamma\beta^{\gamma-1}\sum_{i=1}^nx_i^\gamma.
$$
Let $\frac{\partial l(\alpha,\beta,\gamma)}{\partial\beta}(\alpha^{(t+1)},\beta,\gamma^{(t)})=0$ and then we can get the explicit expression of the MLE of $\beta$ given $\alpha^{(t+1)}$ and $\gamma^{(t)}$:
$$
\beta^{(t+1)}=\left(\frac{n\alpha^{(t+1)}}{\gamma^{(t)}\sum_{i=1}^nx_i^{\gamma^{(t)}}}\right)^{\frac{1}{\gamma^{(t)}}}.
$$

The first- and second-order partial derivatives of $l(\alpha,\beta,\gamma)$ with respect to $\gamma$ are:
\begin{align*}
    \frac{\partial l(\alpha,\beta,\gamma)}{\partial\gamma}=&n\left[\frac{1}{\gamma}+\frac{\alpha}{\gamma^2}\cdot\frac{\Gamma'(\alpha/\gamma)}{\Gamma(\alpha/\gamma)}\right]-\beta^\gamma\sum_{i=1}^nx_i^\gamma\log(\beta x_i)\\
    =&n\left\{\frac{1}{\gamma}+\frac{\alpha}{\gamma^2}\left[-\gamma_0+\sum_{m=0}^\infty\left(\frac{1}{m+1}-\frac{1}{m+\alpha/\gamma}\right)\right]\right\}-\beta^\gamma\sum_{i=1}^nx_i^\gamma\log(\beta x_i)\\
    \frac{\partial^2l(\alpha,\beta,\gamma)}{\partial\gamma^2}=&n\left\{-\frac{1}{\gamma^2}-\frac{2\alpha}{\gamma^3}\left[-\gamma_0+\sum_{m=0}^\infty\left(\frac{1}{m+1}-\frac{1}{m+\alpha/\gamma}\right)\right]-\frac{\alpha^2}{\gamma^4}\sum_{m=0}^\infty\frac{1}{(m+\alpha/\gamma)^2}\right\}\\
    &-\beta^\gamma\sum_{i=1}^nx_i^\gamma[\log(\beta x_i)]^2
\end{align*}
where $\gamma_0=\lim_{n\to\infty}\left(\sum_{m=1}^nm^{-1}-\log n\right)\approx0.5772$ is the Euler-Mascheroni constant.

\begin{align*}
    \frac{\partial^2l(\alpha,\beta,\gamma)}{\partial\gamma^2}=&n\left\{-\frac{1}{\gamma^2}-\frac{2\alpha}{\gamma^3}\left[-\gamma_0+\sum_{m=0}^\infty\left(\frac{1}{m+1}-\frac{1}{m+\alpha/\gamma}\right)\right]-\frac{\alpha^2}{\gamma^4}\sum_{m=0}^\infty\frac{1}{(m+\alpha/\gamma)^2}\right\}\\
    &-\beta^\gamma\sum_{i=1}^nx_i^\gamma[\log(\beta x_i)]^2\\
    \ge&n\left\{-\frac{1}{\gamma^2}-\frac{2\alpha}{\gamma^3}\left[-\gamma_0+\sum_{m=0}^\infty\left(\frac{1}{m+1}-\frac{1}{m+\alpha/\gamma}\right)\right]-\frac{\alpha^2}{\gamma^4}\left(\frac{\gamma^2}{\alpha^2}+\sum_{m=1}^\infty\frac{1}{m^2}\right)\right\}\\
    &-\beta^\gamma\sum_{i=1}^nx_i^\gamma[\log(\beta x_i)]^2\\
    =&n\left\{-\frac{1}{\gamma^2}-\frac{2\alpha}{\gamma^3}\left[-\gamma_0+\sum_{m=0}^\infty\left(\frac{1}{m+1}-\frac{1}{m+\alpha/\gamma}\right)\right]-\frac{\alpha^2}{\gamma^4}\left(\frac{\gamma^2}{\alpha^2}+\frac{\pi^2}{6}\right)\right\}\\
    &-\beta^\gamma\sum_{i=1}^nx_i^\gamma[\log(\beta x_i)]^2\\
    =&n\left\{-\frac{2}{\gamma^2}-\frac{\pi^2}{6}\cdot\frac{\alpha^2}{\gamma^4}-\frac{2\alpha}{\gamma^3}\left[-\gamma_0+\sum_{m=0}^\infty\left(\frac{1}{m+1}-\frac{1}{m+\alpha/\gamma}\right)\right]\right\}\\
    &-\beta^\gamma\sum_{i=1}^nx_i^\gamma[\log(\beta x_i)]^2\\
    =&n\left[-\frac{2}{\gamma^2}+\frac{2\alpha\gamma_0}{\gamma^3}-\frac{\pi^2}{6}\cdot\frac{\alpha^2}{\gamma^4}-\frac{2\alpha}{\gamma^3}\sum_{m=0}^\infty\left(\frac{1}{m+1}-\frac{1}{m+\alpha/\gamma}\right)\right]\\
    &-\beta^\gamma\sum_{i=1}^nx_i^\gamma[\log(\beta x_i)]^2\\
    \ge&n\left[-\frac{2}{\gamma^2}+\frac{2\alpha\gamma_0}{\gamma^3}-\frac{\pi^2}{6}\cdot\frac{\alpha^2}{\gamma^4}-\frac{2\alpha}{\gamma^3}\sum_{m=0}^\infty\left(\frac{1}{m+1}-\frac{1}{m+\lceil\alpha/\gamma\rceil}\right)\right]\\
    &-\beta^\gamma\sum_{i=1}^nx_i^\gamma[\log(\beta x_i)]^2\\
    =&n\left[-\frac{2}{\gamma^2}+\frac{2\alpha\gamma_0}{\gamma^3}-\frac{\pi^2}{6}\cdot\frac{\alpha^2}{\gamma^4}-\frac{2\alpha}{\gamma^3}\left(\sum_{m=0}^\infty\frac{1}{m+1}-\sum_{m=\lceil\alpha/\gamma\rceil-1}^\infty\frac{1}{m+1}\right)\right]\\
    &-\beta^\gamma\sum_{i=1}^nx_i^\gamma[\log(\beta x_i)]^2\\
    =&n\left(-\frac{2}{\gamma^2}+\frac{2\alpha\gamma_0}{\gamma^3}-\frac{\pi^2}{6}\cdot\frac{\alpha^2}{\gamma^4}-\frac{2\alpha}{\gamma^3}\sum_{m=0}^{\lceil\alpha/\gamma\rceil-2}\frac{1}{m+1}\right)-\beta^\gamma\sum_{i=1}^nx_i^\gamma[\log(\beta x_i)]^2\\
    \ge&n\left(-\frac{2}{\gamma^2}+\frac{2\alpha\gamma_0}{\gamma^3}-\frac{\pi^2}{6}\cdot\frac{\alpha^2}{\gamma^4}-\frac{2\alpha}{\gamma^3}\sum_{m=0}^{\lceil\alpha/\gamma\rceil-2}\frac{1}{0+1}\right)-\beta^\gamma\sum_{i=1}^nx_i^\gamma[\log(\beta x_i)]^2\\
    =&n\left[-\frac{2}{\gamma^2}+\frac{2\alpha\gamma_0}{\gamma^3}-\frac{\pi^2}{6}\cdot\frac{\alpha^2}{\gamma^4}-\frac{2\alpha}{\gamma^3}\left(\lceil\frac{\alpha}{\gamma}\rceil-1\right)\right]-\beta^\gamma\sum_{i=1}^nx_i^\gamma[\log(\beta x_i)]^2\\
    \ge&n\left[-\frac{2}{\gamma^2}+\frac{2\alpha\gamma_0}{\gamma^3}-\frac{\pi^2}{6}\cdot\frac{\alpha^2}{\gamma^4}-\frac{2\alpha}{\gamma^3}\left(\frac{\alpha}{\gamma}+1-1\right)\right]-\beta^\gamma\sum_{i=1}^nx_i^\gamma[\log(\beta x_i)]^2\\
    =&n\left(-\frac{2}{\gamma^2}+\frac{2\alpha\gamma_0}{\gamma^3}-(2+\frac{\pi^2}{6})\frac{\alpha^2}{\gamma^4}\right)-\beta^\gamma\sum_{i=1}^nx_i^\gamma[\log(\beta x_i)]^2
\end{align*}
where $\lceil\cdot\rceil$ is the round-up function (ceiling function).

Upon reaching this step, we find that the term $n\left(-\frac{2}{\gamma^2}+\frac{2\alpha\gamma_0}{\gamma^3}-(2+\frac{\pi^2}{6})\frac{\alpha^2}{\gamma^4}\right)$ is easy to handle (by integrating the former term, we will obtain a cubic fractional function with respect to $\gamma$). However, the term $-\beta^\gamma\sum_{i=1}^nx_i^\gamma[\log(\beta x_i)]^2$ is extremely difficult to deal with. Fortunately, upon further observation and previous derivation given by the paper \cite{tian2025second}, we have two inequalities for exponential functions which can be applied to deal with the latter part given by \citeauthor{tian2025second}.

\begin{lemma}[\cite{tian2025second}]\label{exp_inequality}
    Let $\theta\ge0$, we have
    $$
    -e^\theta\ge-\frac{4e^{2\max(\theta^{(t)}-1,0)}}{(2\theta^{(t)}-\theta)^2},\ \forall\ 0\le\theta<2\theta^{(t)}\ {\rm and}\ \theta^{(t)}>0.
    $$
\end{lemma}
\begin{proof}
    Define $g(\theta)\triangleq e^{\theta/2}(2\theta^{(t)}-\theta)$. By setting $0=g'(\theta)=\frac{1}{2}e^{\theta/2}(2\theta^{(t)}-\theta)-e^{\theta/2}$, we obtain its root $\theta^*=2(\theta^{(t)}-1)$. Thus $g(\theta)>0$ if $\theta<\theta^*$ and $g(\theta)<0$ if $\theta>\theta^*$. So, we have
    \begin{align*}
        g(\theta)\le&g(\max(\theta^*,0))=\begin{cases}
            g(\theta^*)=2e^{\theta^{(t)}-1},&{\rm if}\ \theta^*>0,\\
            g(0)=2\theta^{(t)},&{\rm if}\ \theta^*\le0,
        \end{cases}\\
        =&\begin{cases}
            2e^{\theta^{(t)}-1},&{\rm if}\ \theta^{(t)}>1,\\
            2\theta^{(t)},&{\rm if}\ \theta^{(t)}\le1,
        \end{cases}\\
        \le&\begin{cases}
            2e^{\theta^{(t)}-1},&{\rm if}\ \theta^{(t)}>1,\\
            2,&{\rm if}\ \theta^{(t)}\le1,
        \end{cases}\\
        =&2e^{\max(\theta^{(t)}-1,0)}.
    \end{align*}
    So
    \begin{align*}
        e^{\theta/2}(2\theta^{(t)}-\theta)\le&2e^{\max(\theta^{(t)}-1,0)}\\
        \implies e^\theta(2\theta^{(t)}-\theta)^2\le&4e^{2\max(\theta^{(t)}-1,0)}\\
        \implies -e^\theta\ge&-\frac{4e^{2\max(\theta^{(t)}-1,0)}}{(2\theta^{(t)}-\theta)^2}.
    \end{align*}
\end{proof}

\begin{lemma}[\cite{tian2025second}]\label{exp_inequality_2}
    For any $\theta\ge0$, we have
    $$
    -e^{-\theta}\ge-\frac{2}{3}\theta^{-2}.
    $$
\end{lemma}

The proof of Lemma \ref{exp_inequality_2} is given by the paper \cite{tian2025second} in the appendix, and we omit it here.

By Lemma \ref{exp_inequality} and Lemma \ref{exp_inequality_2}, we can further deal with the difficult part of $\frac{\partial^2l(\alpha,\beta,\gamma)}{\partial\gamma^2}$ with some other techniques below:
\begin{align*}
    &-\beta^\gamma\sum_{i=1}^nx_i^\gamma[\log(\beta x_i)]^2\\
    =&-\sum_{i=1}^n(\beta x_i)^\gamma[\log(\beta x_i)]^2\\
    =&-\sum_{i=1}^ne^{\gamma\log(\beta x_i)}[\log(\beta x_i)]^2\\
    \ge&-\sum_{i=1}^n\left\{\frac{4e^{2\max(\gamma^{(t)}\log(\beta x_i)-1,0)}I(\log(\beta x_i)>0)}{[2\gamma^{(t)}\log(\beta x_i)-\gamma\log(\beta x_i)]^2}+\frac{2I(\log(\beta x_i)\le0)}{3[-\gamma\log(\beta x_i)]^2}\right\}[\log(\beta x_i)]^2\\
    =&-\sum_{i=1}^n\left[\frac{4e^{2\max(\gamma^{(t)}\log(\beta x_i)-1,0)}I(x_i>1/\beta)}{(2\gamma^{(t)}-\gamma)^2}+\frac{2I(x_i\le1/\beta)}{3\gamma^2}\right]
\end{align*}
where $I(\cdot)$ is the indicator function. We use Lemma \ref{exp_inequality} by $\forall\ i\in\{1,\cdots,n\}$ letting
\begin{align*}
    \theta_i=&\gamma\log(\beta x_i),\\
    \theta_i^{(t)}=&\gamma^{(t)}\log(\beta x_i),
\end{align*}
so the inequality holds when $0<\gamma<2\gamma^{(t)}$, i.e.,
\begin{align*}
    &0<\gamma\log(\beta x_i)<2\gamma^{(t)}\log(\beta x_i)\\
    \Longleftrightarrow\ &0<\theta_i<2\theta_i^{(t)},
\end{align*}
and we use Lemma \ref{exp_inequality_2} by $\forall\ i\in\{1,\cdots,n\}$ letting
$$
\theta_i=-\gamma\log(\beta x_i).
$$

As a result, we have
\begin{align*}
    \frac{\partial^2l(\alpha,\beta,\gamma)}{\partial\gamma^2}\ge&n\left[-\frac{2}{\gamma^2}+\frac{2\alpha\gamma_0}{\gamma^3}-\left(2+\frac{\pi^2}{6}\right)\frac{\alpha^2}{\gamma^4}\right]\\
    &-\sum_{i=1}^n\left[\frac{4e^{2\max(\gamma^{(t)}\log(\beta x_i)-1,0)}I(x_i>1/\beta)}{(2\gamma^{(t)}-\gamma)^2}+\frac{2I(x_i\le1/\beta)}{3\gamma^2}\right]\\
    \triangleq&b_\gamma(\alpha,\beta,\gamma)
\end{align*}
where $b_\gamma(\alpha,\beta,\gamma)$ is a lower bound of the second-order partial derivative $\frac{\partial^2l(\alpha,\beta,\gamma)}{\partial\gamma^2}$.

Then we can construct a $U$-function of $\frac{\partial l(\alpha,\beta,\gamma)}{\partial\gamma}$ by $b_\gamma(\alpha,\beta,\gamma)$ as follows:
\begin{align*}
    &U_\gamma(\alpha,\beta,\gamma|\alpha^{(t+1)},\beta^{(t+1)},\gamma^{(t)})\\
    =&\frac{\partial l(\alpha,\beta,\gamma)}{\partial\gamma}(\alpha^{(t+1)},\beta^{(t+1)},\gamma^{(t)})+\int_{\gamma^{(t)}}^\gamma b_\gamma(\alpha^{(t+1)},\beta^{(t+1)},z){\rm d}z\\
    =&n\left\{\frac{1}{\gamma^{(t)}}+\frac{\alpha^{(t+1)}}{(\gamma^{(t)})^2}\left[-\gamma_0+\sum_{m=0}^\infty\left(\frac{1}{m+1}-\frac{1}{m+\alpha^{(t+1)}/\gamma^{(t)}}\right)\right]\right\}\\
    &-(\beta^{(t+1)})^{\gamma^{(t)}}\sum_{i=1}^nx_i^{\gamma^{(t)}}\log(\beta^{(t+1)}x_i)\\
    &+\Bigg\{n\left[\frac{2}{z}-\frac{\alpha^{(t+1)}\gamma_0}{z^2}+\left(\frac{2}{3}+\frac{\pi^2}{18}\right)\frac{(\alpha^{(t+1)})^2}{z^3}\right]\\
    &-\sum_{i=1}^n\left[\frac{4e^{2\max(\gamma^{(t)}\log(\beta^{(t+1)}x_i)-1,0)}I(x_i>1/\beta^{(t+1)})}{2\gamma^{(t)}-z}-\frac{2I(x_i\le1/\beta^{(t+1)})}{3z}\right]\Bigg\}_{\gamma^{(t)}}^\gamma\\
    =&n\left\{\frac{1}{\gamma^{(t)}}+\frac{\alpha^{(t+1)}}{(\gamma^{(t)})^2}\left[-\gamma_0+\sum_{m=0}^\infty\left(\frac{1}{m+1}-\frac{1}{m+\alpha^{(t+1)}/\gamma^{(t)}}\right)\right]\right\}\\
    &-(\beta^{(t+1)})^{\gamma^{(t)}}\sum_{i=1}^nx_i^{\gamma^{(t)}}\log(\beta^{(t+1)}x_i)\\
    &-n\left[\frac{2}{\gamma^{(t)}}-\frac{\alpha^{(t+1)}\gamma_0}{(\gamma^{(t)})^2}+\left(\frac{2}{3}+\frac{\pi^2}{18}\right)\frac{(\alpha^{(t+1)})^2}{(\gamma^{(t)})^3}\right]\\
    &+\sum_{i=1}^n\left[\frac{4e^{2\max(\gamma^{(t)}\log(\beta^{(t+1)}x_i)-1,0)}I(x_i>1/\beta^{(t+1)})}{\gamma^{(t)}}-\frac{2I(x_i\le1/\beta^{(t+1)})}{3\gamma^{(t)}}\right]\\
    &+n\left[\frac{2}{\gamma}-\frac{\alpha^{(t+1)}\gamma_0}{\gamma^2}+\left(\frac{2}{3}+\frac{\pi^2}{18}\right)\frac{(\alpha^{(t+1)})^2}{\gamma^3}\right]\\
    &-\sum_{i=1}^n\left[\frac{4e^{2\max(\gamma^{(t)}\log(\beta^{(t+1)}x_i)-1,0)}I(x_i>1/\beta^{(t+1)})}{2\gamma^{(t)}-\gamma}-\frac{2I(x_i\le1/\beta^{(t+1)})}{3\gamma}\right].
\end{align*}

Letting $U_\gamma(\alpha,\beta,\gamma|\alpha^{(t+1)},\beta^{(t+1)},\gamma^{(t)})=0$, we get
$$
\gamma^{(t+1)}={\rm sol}\left\{0<\gamma<2\gamma^{(t)}:d_4^{(t)}\gamma^4+d_3^{(t)}\gamma^3+d_2^{(t)}\gamma^2+d_1^{(t)}\gamma+d_0^{(t)}=0\right\}
$$
where
\begin{align*}
    d_4^{(t)}=&-n\left\{\frac{1}{\gamma^{(t)}}+\frac{\alpha^{(t+1)}}{(\gamma^{(t)})^2}\left[-\gamma_0+\sum_{m=0}^\infty\left(\frac{1}{m+1}-\frac{1}{m+\alpha^{(t+1)}/\gamma^{(t)}}\right)\right]\right\}\\
    &+(\beta^{(t+1)})^{\gamma^{(t)}}\sum_{i=1}^nx_i^{\gamma^{(t)}}\log(\beta^{(t+1)}x_i)\\
    &+n\left[\frac{2}{\gamma^{(t)}}-\frac{\alpha^{(t+1)}\gamma_0}{(\gamma^{(t)})^2}+\left(\frac{2}{3}+\frac{\pi^2}{18}\right)\frac{(\alpha^{(t+1)})^2}{(\gamma^{(t)})^3}\right]\\
    &-\sum_{i=1}^n\left[\frac{4e^{2\max(\gamma^{(t)}\log(\beta^{(t+1)}x_i)-1,0)}I(x_i>1/\beta^{(t+1)})}{\gamma^{(t)}}-\frac{2I(x_i\le1/\beta^{(t+1)})}{3\gamma^{(t)}}\right],\\
    d_3^{(t)}=&2\gamma^{(t)}\Bigg(n\left\{\frac{1}{\gamma^{(t)}}+\frac{\alpha^{(t+1)}}{(\gamma^{(t)})^2}\left[-\gamma_0+\sum_{m=0}^\infty\left(\frac{1}{m+1}-\frac{1}{m+\alpha^{(t+1)}/\gamma^{(t)}}\right)\right]\right\}\\
    &-(\beta^{(t+1)})^{\gamma^{(t)}}\sum_{i=1}^nx_i^{\gamma^{(t)}}\log(\beta^{(t+1)}x_i)\\
    &-n\left[\frac{2}{\gamma^{(t)}}-\frac{\alpha^{(t+1)}\gamma_0}{(\gamma^{(t)})^2}+\left(\frac{2}{3}+\frac{\pi^2}{18}\right)\frac{(\alpha^{(t+1)})^2}{(\gamma^{(t)})^3}\right]\\
    &+\sum_{i=1}^n\left[\frac{4e^{2\max(\gamma^{(t)}\log(\beta^{(t+1)}x_i)-1,0)}I(x_i>1/\beta^{(t+1)})}{\gamma^{(t)}}-\frac{2I(x_i\le1/\beta^{(t+1)})}{3\gamma^{(t)}}\right]\Bigg)\\
    &-2n-\frac{2}{3}\sum_{i=1}^nI\left(x_i\le\frac{1}{\beta^{(t+1)}}\right)-\sum_{i=1}^n4e^{2\max(\gamma^{(t)}\log(\beta^{(t+1)}x_i)-1,0)}I\left(x_i>\frac{1}{\beta^{(t+1)}}\right),\\
    d_2^{(t)}=&2\gamma^{(t)}\left[2n+\frac{2}{3}\sum_{i=1}^nI\left(x_i\le\frac{1}{\beta^{(t+1)}}\right)\right]+\gamma_0n\alpha^{(t+1)},\\
    d_1^{(t)}=&-2\gamma_0n\alpha^{(t+1)}\gamma^{(t)}-\left(\frac{2}{3}+\frac{\pi^2}{18}\right)n(\alpha^{(t+1)})^2,\\
    d_0^{(t)}=&\left(\frac{4}{3}+\frac{\pi^2}{9}\right)n\gamma^{(t)}(\alpha^{(t+1)})^2.
\end{align*}

According to the Galois theory which is introduced in great detail by the book \cite{morandi2012field} in abstract algebra, the Abel–Ruffini theorem states that general polynomial equations of degree five or higher do not have explicit radical formulas for finding their roots, which has many different types of proofs and one of them is given by the paper \cite{skopenkov2011simple}. The equation we have here is a univariate quartic equation, which is the equation of the highest degree that has an analytical explicit radical root formula. This implies, in a sense, that our lower bound function $b_\gamma(\alpha,\beta,\gamma)$ is very tight.

We now present the pseudo code for each iteration of our SeLF (or US) algorithm, which estimates the parameters $\alpha$, $\beta$, and $\gamma$ as follows:
\begin{table*}[h!]
    \begin{center}
        \begin{tabularx}{\textwidth}{lXXX}
            \textbf{Step 1:}&$\alpha^{(t+1)}={\rm sol}\left\{\alpha>0:a^{(t)}\alpha^2+b^{(t)}\alpha+c^{(t)}=0\right\}$;\\
            \textbf{Step 2:}&$\beta^{(t+1)}=\left(\frac{n\alpha^{(t+1)}}{\gamma^{(t)}\sum_{i=1}^nx_i^{\gamma^{(t)}}}\right)^{\frac{1}{\gamma^{(t)}}}$;\\
            \textbf{Step 3:}&$\gamma^{(t+1)}={\rm sol}\left\{0<\gamma<2\gamma^{(t)}:d_4^{(t)}\gamma^4+d_3^{(t)}\gamma^3+d_2^{(t)}\gamma^2+d_1^{(t)}\gamma+d_0^{(t)}=0\right\}$.
        \end{tabularx}
    \end{center}
\end{table*}
\FloatBarrier

\section{Experiments}

We have completed the SeLF (or US) algorithm designed for estimating the parameters of the generalized Gamma distribution. The next step is to perform simulation experiments to verify the feasibility of this algorithm and compare its speed performance with the traditional Newton algorithm. The first step in conducting these experiments is to find an algorithm that can generate random variables from the generalized Gamma distribution. Since the generalized Gamma distribution is not a commonly used distribution, and we cannot directly use library functions to generate it, we need to design our own algorithm to generate it.

We will use the Sampling/Importance Resampling (SIR) method which is proposed by \citeauthor{rubin1987calculation}'s paper \cite{rubin1987calculation} and is introduced in detail by the book \cite{tian2023comp} to generate the generalized Gamma distribution, whose probability density function is given by $f(x)=\frac{\beta^\alpha\gamma}{\Gamma(\alpha/\gamma)}x^{\alpha-1}e^{-(\beta x)^\gamma}$ for $x>0$, where $\alpha$, $\beta$, and $\gamma$ are positive parameters. It is very difficult to generate a sample from the original probability density function $f(x)$. For fitting, we will use the probability density function of the traditional Gamma distribution, $g(x)=\frac{\beta^\alpha}{\Gamma(\alpha)}x^{\alpha-1}e^{-\beta x}$ for $x>0$, as our importance resampling density (also known as the proposal density) from which it is relatively easy to draw a sample. First, we could generate $X^{(1)},\cdots,X^{(J)}\overset{{\rm iid}}{\sim}g(x)=\frac{\beta^\alpha}{\Gamma(\alpha)}x^{\alpha-1}e^{-\beta x},\ x>0$ We then define:
\begin{align*}
    w(X^{(j)})\triangleq&\frac{f(X^{(j)})}{g(X^{(j)})}\\
    =&\frac{\Gamma(\alpha)\gamma}{\Gamma(\alpha/\gamma)}\exp\left[\beta X^{(j)}-(\beta X^{(j)})^\gamma\right]
\end{align*}
for all $j\in\left\{1,\cdots,J\right\}$. The second step is to adjust the generated samples $\{X^{(j)}\}_{j=1}^J$ so that part of them becomes samples from $f(x)$ by comparing their ratios $w(X^{(j)})=\frac{f(X^{(j)})}{g(X^{(j)})}$. We can calculate the probability weights by
$$
\omega_j=\frac{w(X^{(j)})}{\sum_{j'=1}^Jw(X^{(j')})},\ j=1,\cdots,J.
$$

So we can give the pseudo code of the SIR method without replacement to generate $I$ samples from the generalized Gamma distribution with parameters $\alpha,\beta,\gamma$.
\begin{table*}[h!]
    \begin{center}
        \begin{tabularx}{\textwidth}{lXXX}
            \textbf{Step 1:}&Generate $X^{(1)},\cdots,X^{(J)}\overset{{\rm iid}}{\sim}{\rm Gamma}(\alpha,\beta)$;\\
            \textbf{Step 2:}&Select a subset $\{X^{(k_i)}\}_{i=1}^I$ from $\{X^{(j)}\}_{j=1}^J$ via re-sapling \textit{without replacement} from the discrete distribution on $\{X^{(j)}\}_{j=1}^J$ with probabilities $\{\omega_j\}_{j=1}^J$.
        \end{tabularx}
    \end{center}
\end{table*}
\FloatBarrier

The SIR algorithm is exact as $J/I\to\infty$, and the larger the discrepancy between the importance resampling function $g(x)$ and the target distribution $f(x)$, the greater the required value of $J$ relative to $I$. This is why we have chosen the traditional Gamma distribution to serve as the importance resampling function. In practice, it is recommended that $J/I\ge10$, with a suggested value of $J/I=20$. \citeauthor{smith1992bayesian} recommended $J/I\ge10$ in their examples of the paper \cite{smith1992bayesian}.

We need another algorithm to compare with our derived SeLF algorithm, and for this, we can choose Newton's method. The idea behind Newton's method is to use the ``straight line approximation" strategy to approach the zero points of a curve step by step, while our US algorithm (or SeLF algorithm) uses a simpler curve to approximate the zero points of the original curve. Therefore, in this case, Newton's method is a more suitable comparison algorithm. The detailed discussion of Newton's method is in the book \cite{tian2023comp}. We now give the pseudo code of each iteration of using Newton's method to estimate the parameters of the generalized Gamma distribution below.

\begin{table*}[h!]
    \begin{center}
        \begin{tabularx}{\textwidth}{lXXX}
            \textbf{Step 1:}&$\alpha^{(t+1)}=\alpha^{(t)}-\frac{\partial l(\alpha,\beta,\gamma)}{\partial\alpha}(\alpha^{(t)},\beta^{(t)},\gamma^{(t)})/\frac{\partial^2l(\alpha,\beta,\gamma)}{\partial\alpha^2}(\alpha^{(t)},\beta^{(t)},\gamma^{(t)})$;\\
            \textbf{Step 2:}&$\beta^{(t+1)}=\left(\frac{n\alpha^{(t+1)}}{\gamma^{(t)}\sum_{i=1}^nx_i^{\gamma^{(t)}}}\right)^{\frac{1}{\gamma^{(t)}}}$;\\
            \textbf{Step 3:}&$\gamma^{(t+1)}=\gamma^{(t)}-\frac{\partial l(\alpha,\beta,\gamma)}{\partial\gamma}(\alpha^{(t+1)},\beta^{(t+1)},\gamma^{(t)})/\frac{\partial^2l(\alpha,\beta,\gamma)}{\partial\gamma^2}(\alpha^{(t+1)},\beta^{(t+1)},\gamma^{(t)})$.
        \end{tabularx}
    \end{center}
\end{table*}
\FloatBarrier

Since the MLE of $\beta$ has an explicit expression given $\alpha$ and $\beta$, we do not need to use Newton's method on it, which is similar to the SeLF (or US) algorithm we put forward previously.

Our experiments first generate 1000 independent identically distributed (i.i.d.) samples $\{x_i\}_{i=1}^{1000}$ from the generalized Gamma distribution with parameters $\alpha=2$, $\beta=3$ and $\gamma=2$ and set the initial values of both the SeLF algorithm and the Newton's method as $\alpha^{(0)}=3$, $\beta^{(0)}=2$ and $\gamma^{(0)}=3$. We iterate 200 times using both the SeLF algorithm and Newton's method, respectively.

Our experimental results consist of a total of 201 rows and 6 columns. The first row shows the initial values $(\alpha^{(0)},\beta^{(0)},\gamma^{(0)})$ of the SeLF algorithm and the Newton's method. Each of the remaining rows represents the maximum likelihood estimate given by the SeLF algorithm and Newton's method after increasing the number of iterations by one compared to the previous row. The first three columns correspond to the parameter estimates given by the SeLF algorithm, while the last three columns correspond to the parameter estimates given by Newton's method. For ease of reading, we only present the first row and the rows corresponding to iteration counts that are multiples of ten.

\begin{table*}[h!]
    \begin{center}
        \begin{tabular}{|c|c|c|c|c|c|c|}
            \hline
            Iteration&\multicolumn{3}{c|}{SeLF}&\multicolumn{3}{c|}{Newton}\\
            Counts&$\alpha$&$\beta$&$\gamma$&$\alpha$&$\beta$&$\gamma$\\
            \hline
            0   & 3         & 2         & 3         & 3         & 2         & 3         \\ \hline
1   & 1.6681069 & 2.2756410 & 2.9668302 & 0.2102209 & 1.1409181 & 2.3397234 \\ \hline
10  & 1.6963639 & 2.3615199 & 2.7584139 & 2.3638493 & 5.6555624 & 1.2432267 \\ \hline
20  & 1.7477509 & 2.4524057 & 2.6021852 & 2.7105817 & 5.8654599 & 1.2916193 \\ \hline
30  & 1.7886466 & 2.5286310 & 2.4910096 & 2.6820881 & 5.6478946 & 1.3174232 \\ \hline
40  & 1.8216708 & 2.5928895 & 2.4084698 & 2.6444994 & 5.4290938 & 1.3441471 \\ \hline
50  & 1.8487115 & 2.6473951 & 2.3451589 & 2.6062957 & 5.2185234 & 1.3723236 \\ \hline
60  & 1.8711203 & 2.6939061 & 2.2953465 & 2.5677527 & 5.0167583 & 1.4019825 \\ \hline
70  & 1.8898812 & 2.7338131 & 2.2553567 & 2.5290153 & 4.8242072 & 1.4331187 \\ \hline
80  & 1.9057248 & 2.7682211 & 2.2227284 & 2.4902527 & 4.6412652 & 1.4656924 \\ \hline
90  & 1.9192033 & 2.7980141 & 2.1957538 & 2.4516625 & 4.4683106 & 1.4996205 \\ \hline
100 & 1.9307412 & 2.8239062 & 2.1732106 & 2.4134715 & 4.3056967 & 1.5347656 \\ \hline
110 & 1.9406705 & 2.8464803 & 2.1542007 & 2.3759343 & 4.1537410 & 1.5709268 \\ \hline
120 & 1.9492543 & 2.8662162 & 2.1380490 & 2.3393297 & 4.0127110 & 1.6078316 \\ \hline
130 & 1.9567038 & 2.8835125 & 2.1242379 & 2.3039539 & 3.8828069 & 1.6451327 \\ \hline
140 & 1.9631908 & 2.8987029 & 2.1123639 & 2.2701101 & 3.7641422 & 1.6824113 \\ \hline
150 & 1.9688562 & 2.9120685 & 2.1021077 & 2.2380941 & 3.6567225 & 1.7191897 \\ \hline
160 & 1.9738167 & 2.9238477 & 2.0932134 & 2.2081776 & 3.5604265 & 1.7549538 \\ \hline
170 & 1.9781696 & 2.9342435 & 2.0854736 & 2.1805901 & 3.4749905 & 1.7891859 \\ \hline
180 & 1.9819968 & 2.9434298 & 2.0787181 & 2.1555018 & 3.4000017 & 1.8214031 \\ \hline
190 & 1.9853674 & 2.9515561 & 2.0728065 & 2.1330106 & 3.3349017 & 1.8511961 \\ \hline
200 & 1.9883404 & 2.9587517 & 2.0676217 & 2.1131357 & 3.2790025 & 1.8782621 \\ \hline
        \end{tabular}
    \end{center}
\end{table*}
\FloatBarrier

All of the results of our experiments are shown in the following figure.

\begin{figure}[h!]
    \centering
    \includegraphics{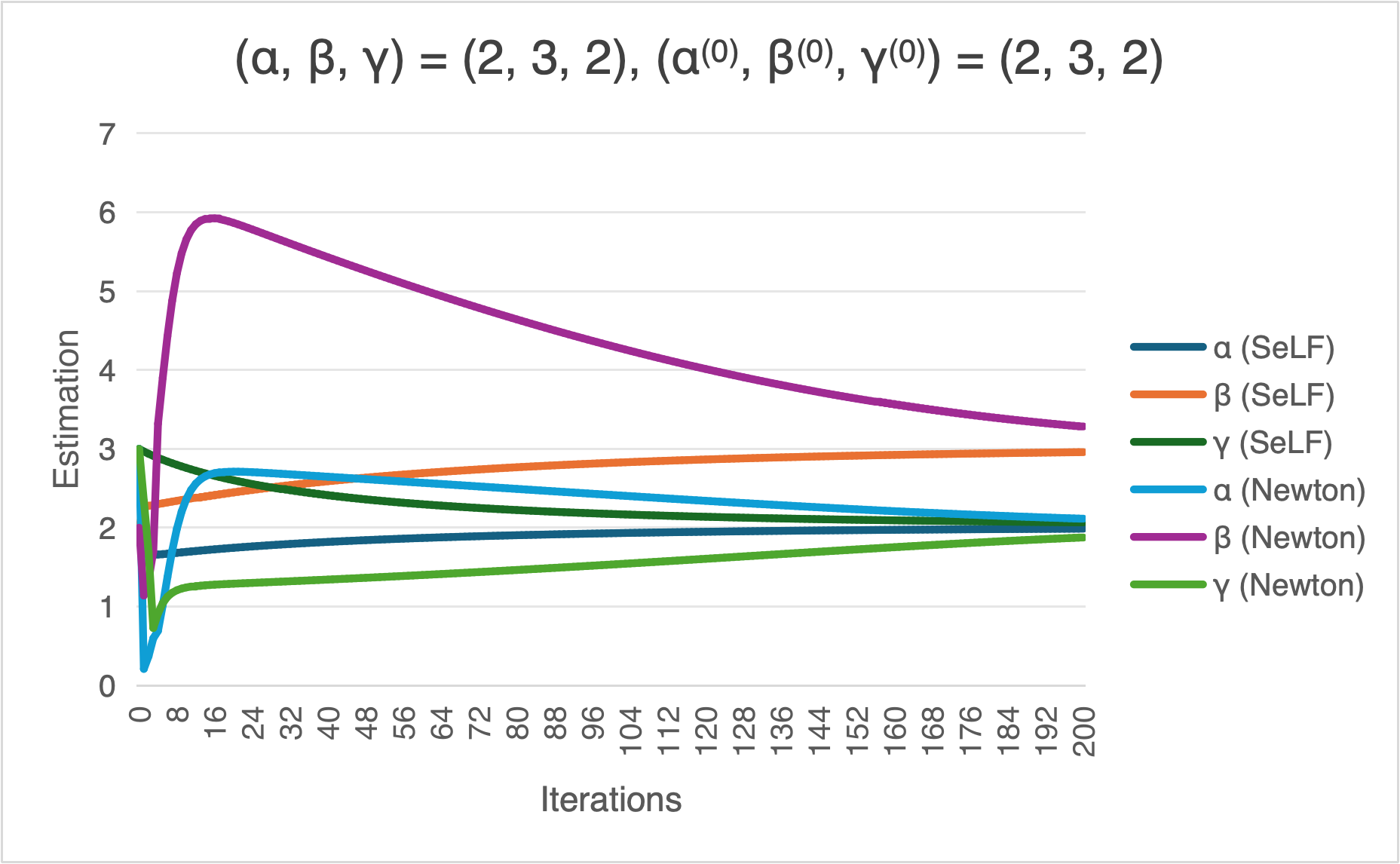}
    \caption{Simulation results of both algorithms.}
\end{figure}
\FloatBarrier

The above chart shows the changes in the parameters $\alpha$, $\beta$, and $\gamma$ estimated by our SeLF algorithm and Newton's method as the number of iterations increases, but it is difficult to observe due to the large number of curves displayed (a total of six). Therefore, we have drawn two separate charts to display the estimation results of the SeLF algorithm and Newton's method, as shown below.

\begin{figure}[h!]
    \centering
    \includegraphics{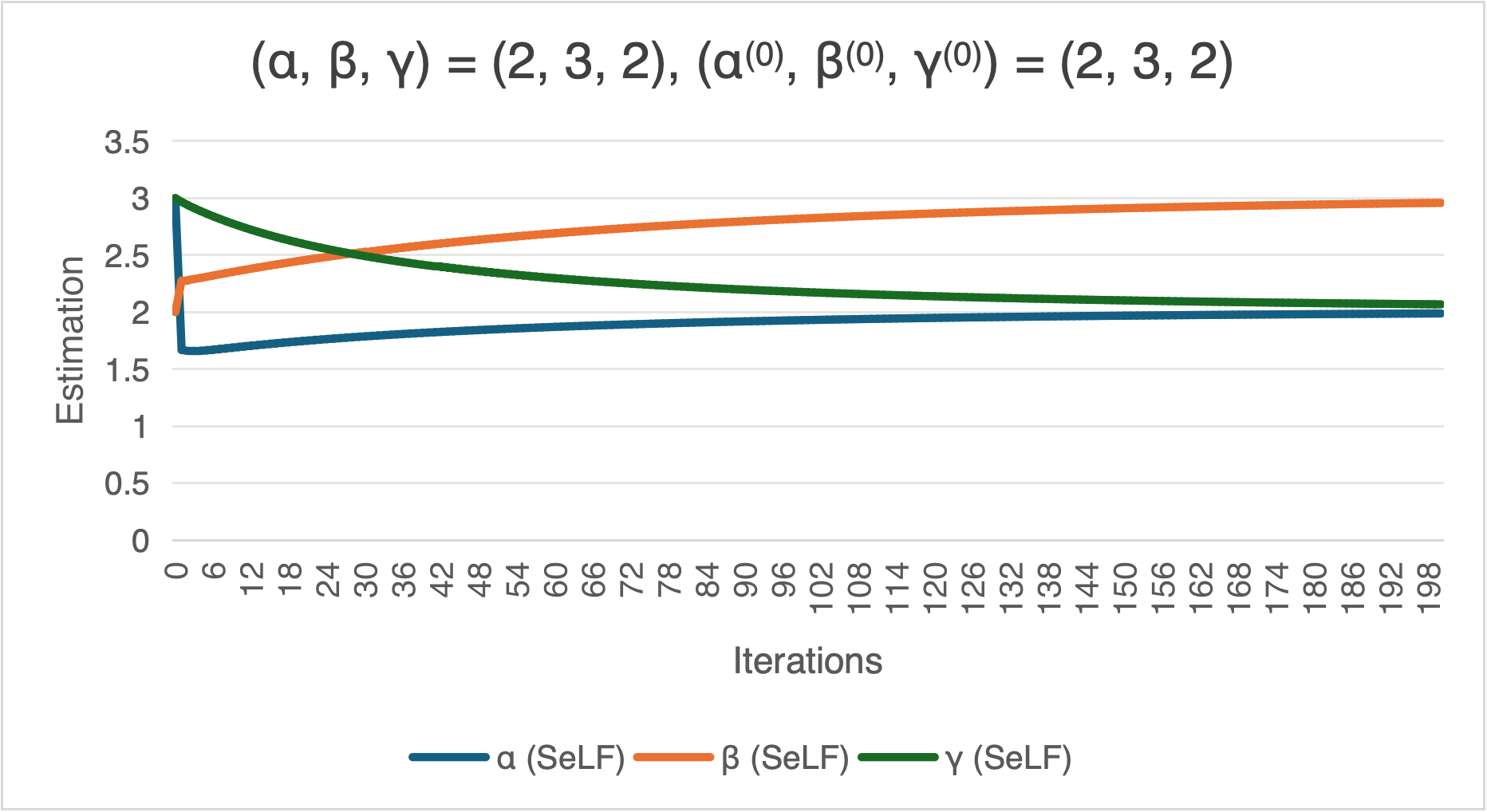}
    \caption{Simulation results of the SeLF algorithm.}
\end{figure}

\begin{figure}[h!]
    \centering
    \includegraphics{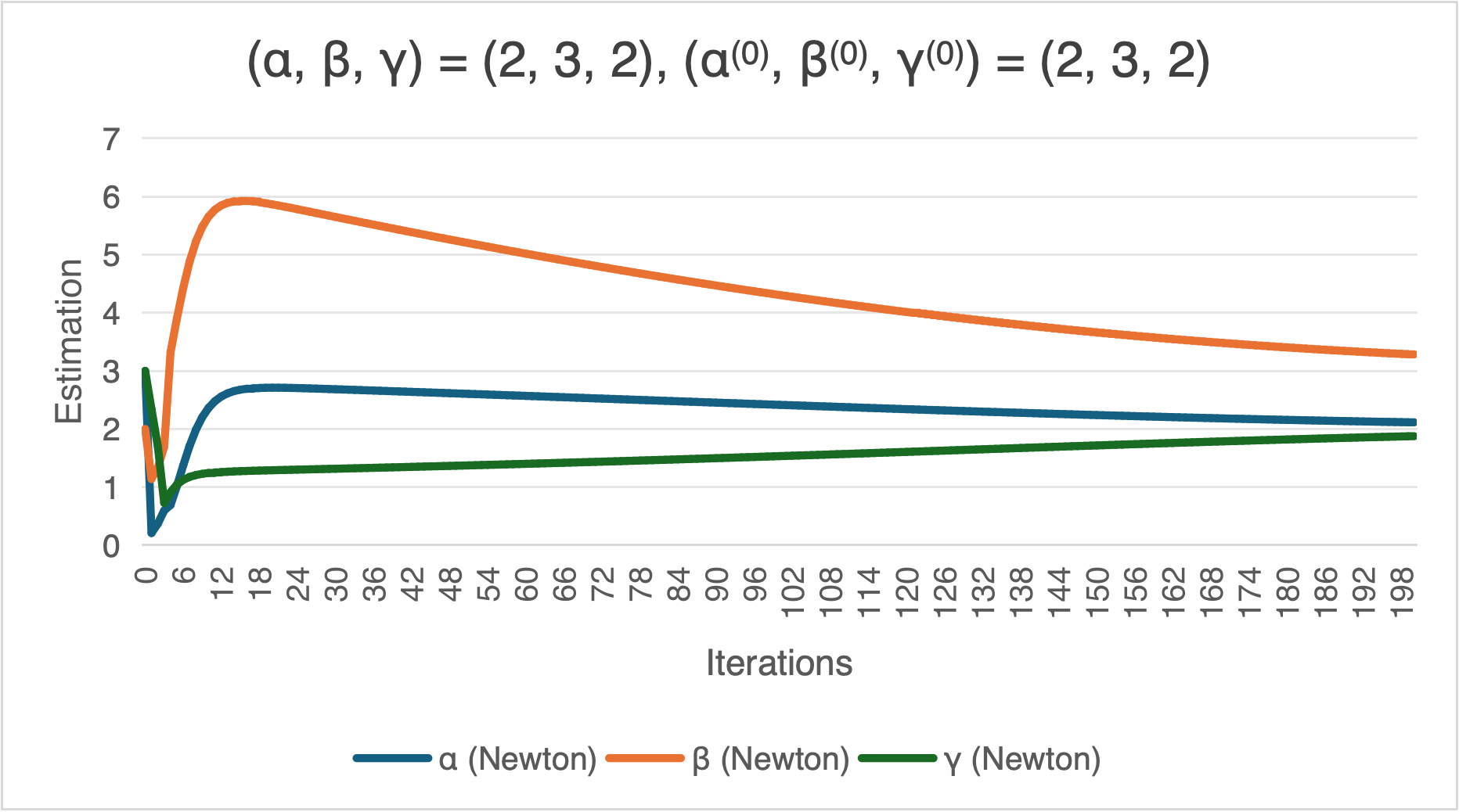}
    \caption{Simulation results of Newton's method.}
\end{figure}
\FloatBarrier

To facilitate a separate comparison of the trends in the changes of each parameter ($\alpha$, $\beta$, and $\gamma$) in the SeLF algorithm and Newton's method, we have created another three individual charts displaying the trends in the changes of each parameter as follows.

\begin{figure}[h!]
    \centering
    \subfigure[Compare estimates of $\alpha$.]
    {
        \includegraphics[width=0.47\columnwidth]{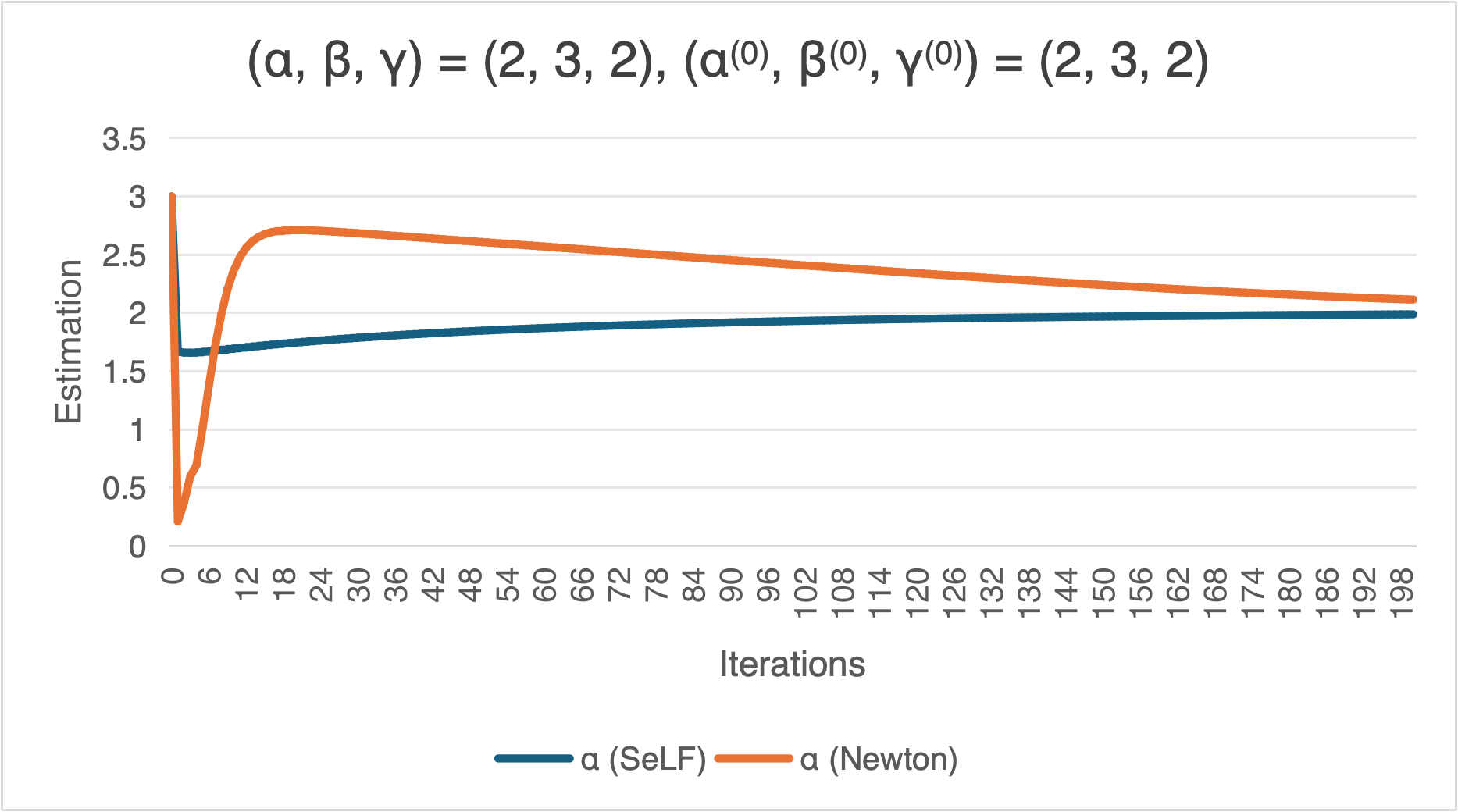}
    }
    \subfigure[Compare estimates of $\beta$.]
    {
        \includegraphics[width=0.47\columnwidth]{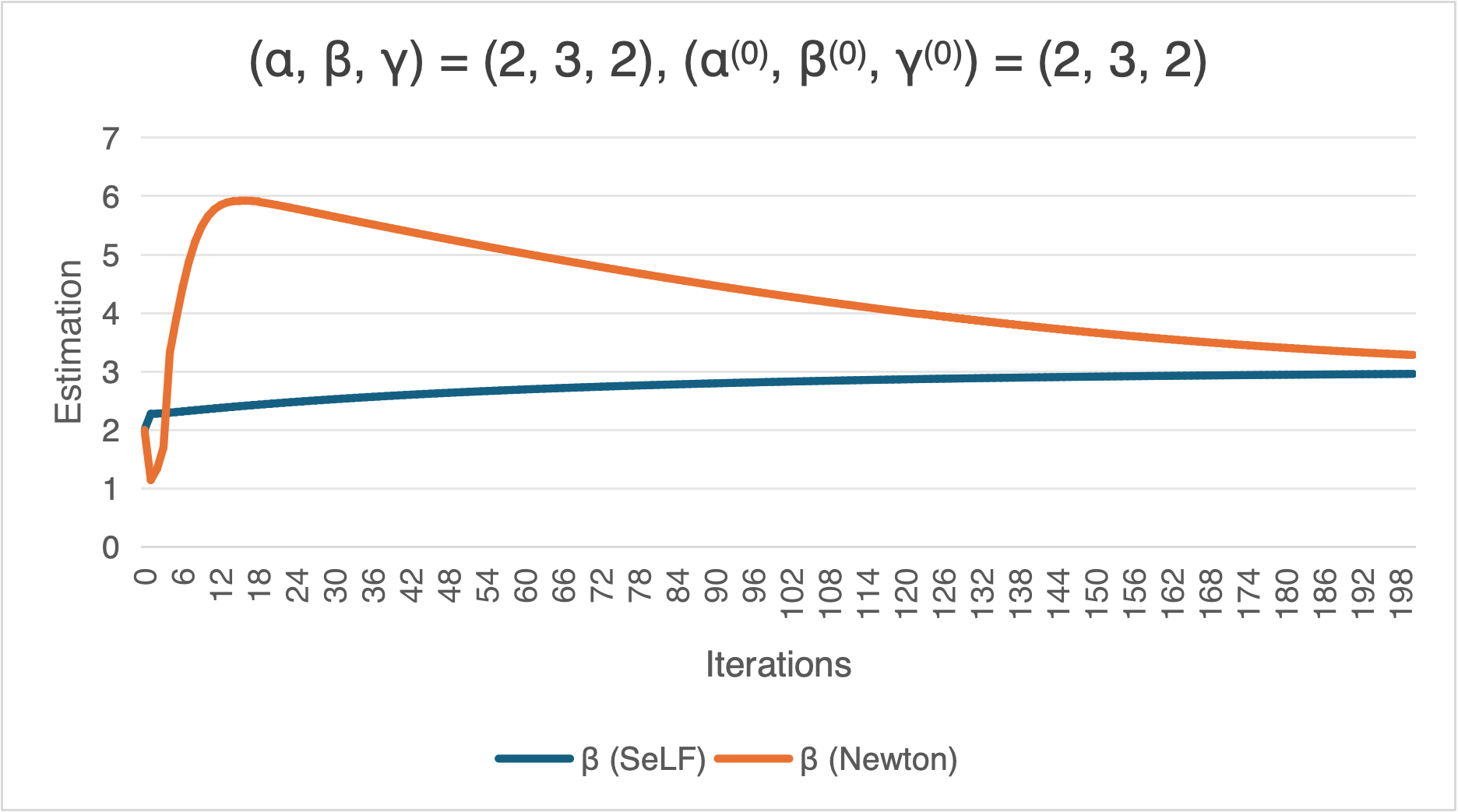}
    }
    \subfigure[Compare estimates of $\gamma$.]
    {
        \includegraphics[width=0.47\columnwidth]{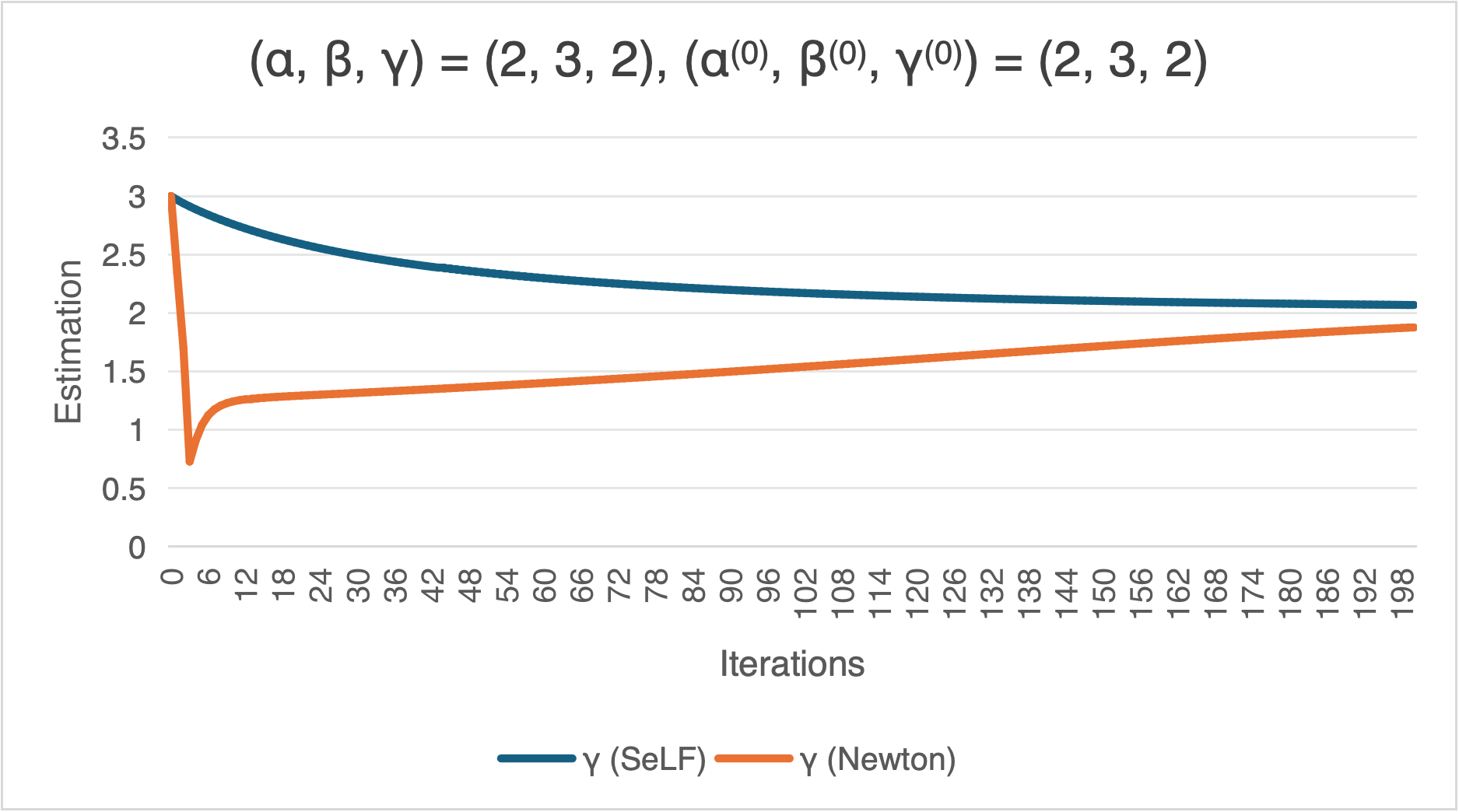}
    }
    \caption{Compare estimates of each parameter. (The blue curves display the estimates of the SeLF algorithm while the orange curves display the estimates of Newton's method.)}
\end{figure}
\FloatBarrier

By observing the experimental results above, we can comprehensively compare the advantages and disadvantages of the SeLF algorithm and Newton's method for maximum likelihood estimation of the parameters in the generalized Gamma distribution in terms of stability, convergence speed, algorithm complexity, and estimation accuracy. These advantages and disadvantages are described as follows.

\paragraph{Stability}

In terms of stability, the SeLF algorithm is more stable than Newton's method. It is observed that the SeLF algorithm usually only experiences a reverse jump in the first iteration, and thereafter it monotonically increases or decreases to approach the true values of the parameters ($\alpha$, $\beta$, and $\gamma$). In contrast, Newton's method exhibits greater fluctuation during iterations and tends to have more turning points. It typically takes two or more reversals before it achieves a monotonic trend towards the true values of the parameters.

\paragraph{Convergence Speed}

In terms of convergence speed, under the same number of iterations, the SeLF algorithm has a faster convergence speed than Newton's method. The maximum likelihood estimates obtained by the SeLF algorithm are generally closer to the true values than those obtained by Newton's method, and this tends to occur after a high number of iterations.

\paragraph{Algorithm Complexity}

In terms of algorithm complexity, the SeLF algorithm requires more time for each iteration than Newton's method because it involves solving a quartic equation and calculating more summations related to the data (the i.i.d. samples $\{x_i\}_{i=1}^{1000}$). Therefore, under the same number of iterations, Newton's method takes much less time than the SeLF algorithm. However, this difference is only a constant factor, and it is not the focus of our research.

\paragraph{Accuracy}

In terms of accuracy, under the same number of iterations, the SeLF algorithm provides maximum likelihood estimates that are closer to the true values than those obtained by Newton's method, particularly when the number of iterations is relatively high. In contrast, Newton's method tends to produce estimates that are farther away from the true values.

\section{Conclusion}

In summary, this undergraduate thesis presents a thorough investigation of the application of the SeLF algorithm (or US algorithm) to determine the maximum likelihood estimation of the parameters, namely $\alpha$, $\beta$, and $\gamma$, in the generalized Gamma distribution. The research begins with a brief introduction to the generalized Gamma distribution and its real-world data fitting application, followed by a concise overview of the MM algorithm, the SeLF algorithm, and the US algorithm. By implementing the SeLF algorithm (or US algorithm), this study proposes a novel method for calculating maximum likelihood estimates of the parameters in the generalized Gamma distribution, which demonstrates superior performance in terms of stability, accuracy, and convergence speed. The results reveal that the SeLF algorithm surpasses traditional methods, such as Newton's method. The proposed method offers practical significance and application value for resolving practical problems related to generalized Gamma distributions. Overall, this thesis contributes to the advancement of statistical methods for parameter estimation in intricate models and provides a reference for further exploration of the theory and application of the generalized Gamma distribution, the SeLF algorithm, and the US algorithm.

\printbibliography

@article{stacy1962generalization,
  title={A generalization of the gamma distribution},
  author={Stacy, Edney W},
  journal={The Annals of mathematical statistics},
  pages={1187--1192},
  year={1962},
  publisher={JSTOR}
}

@article{tian2025second,
  title={The second--derivative lower--bound function (SeLF) algorithm and three acceleration techniques for maximization with strongly stable convergence},
  author={Tian, Guo-Liang and Zhou, Hua and Lange, Kenneth and Li, Xun-Jian},
  journal={Statistics and Computing},
  volume={35},
  number={4},
  pages={113},
  year={2025},
  publisher={Springer}
}

@article{li2022upper,
  title={The upper-crossing/solution (US) algorithm for root-finding with strongly stable convergence},
  author={LI, Xunjian and Tian, Guo-Liang},
  journal={arXiv preprint arXiv:2212.00797},
  year={2022}
}

@book{tian2023comp,
  author={Guo-Liang TIAN}, 
  title={Computational Statistics},
  publisher={Science Press},
  year=2023,
  edition=1,
  month=3,
  isbn={9787030731890}
}

@article{shanker2016modeling,
  title={On modeling of lifetime data using three-parameter generalized Lindley and generalized gamma distributions},
  author={Shanker, R and Shukla, KK},
  journal={Biometrics \& biostatistics International journal},
  volume={4},
  number={7},
  pages={1--7},
  year={2016}
}

@article{mert2015statistical,
  title={A statistical analysis of wind speed data using Burr, generalized gamma, and Weibull distributions in Antakya, Turkey},
  author={Mert, Ilker and Karaku{\c{s}}, Cuma},
  journal={Turkish Journal of Electrical Engineering and Computer Sciences},
  volume={23},
  number={6},
  pages={1571--1586},
  year={2015}
}

@inproceedings{gomes2006four,
  title={Four-parameter generalized gamma distribution used for stock return modelling},
  author={Gom{\`e}s, Oph{\'e}lie and Combes, Catherine and Dussauchoy, Alain},
  booktitle={The Proceedings of the Multiconference on" Computational Engineering in Systems Applications"},
  volume={1},
  pages={380--386},
  year={2006},
  organization={IEEE}
}

@article{hunter2004tutorial,
  title={A tutorial on MM algorithms},
  author={Hunter, David R and Lange, Kenneth},
  journal={The American Statistician},
  volume={58},
  number={1},
  pages={30--37},
  year={2004},
  publisher={Taylor \& Francis}
}

@book{morandi2012field,
  title={Field and Galois theory},
  author={Morandi, Patrick},
  volume={167},
  year={2012},
  publisher={Springer Science \& Business Media}
}

@article{skopenkov2011simple,
  title={A simple proof of the Abel-Ruffini theorem},
  author={Skopenkov, A},
  journal={arXiv preprint arXiv:1102.2100},
  year={2011}
}

@article{rubin1987calculation,
  title={The calculation of posterior distributions by data augmentation: Comment: A noniterative sampling/importance resampling alternative to the data augmentation algorithm for creating a few imputations when fractions of missing information are modest: The SIR algorithm},
  author={Rubin, Donald B},
  journal={Journal of the American Statistical Association},
  volume={82},
  number={398},
  pages={543--546},
  year={1987},
  publisher={JSTOR}
}

@article{smith1992bayesian,
  title={Bayesian statistics without tears: a sampling--resampling perspective},
  author={Smith, Adrian FM and Gelfand, Alan E},
  journal={The American Statistician},
  volume={46},
  number={2},
  pages={84--88},
  year={1992},
  publisher={Taylor \& Francis}
}

\section*{Appendix}
\subsection*{Experiment Code}

\begin{lstlisting}[style=python]
import math
import numpy as np
import sympy as sp

gamma_0 = 0.5772156649015328606065120900824024310421593359399235988057672348
J_div_I = 30
infinite_sum_upper_calculation_times = 1000


def generalized_gamma(alpha, beta, gamma, size):
    x = np.random.gamma(alpha, 1 / beta, size * J_div_I)
    w = (
        math.gamma(alpha)
        * gamma
        / math.gamma(alpha / gamma)
        * np.exp(beta * x - (beta * x) ** gamma)
    )
    w /= np.sum(w)
    return np.random.choice(x, size, False, w)


def solve_quadratic_equation(a, b, c):
    x = sp.Symbol("x")
    f = a * x**2 + b * x + c
    return sp.solve(f)


def solve_quartic_equation(d_4, d_3, d_2, d_1, d_0):
    x = sp.Symbol("x")
    f = d_4 * x**4 + d_3 * x**3 + d_2 * x**2 + d_1 * x + d_0
    return sp.solve(f)


def get_positive_largest(x, upper_bound=math.inf):
    result = 0
    for element in x:
        try:
            if element < upper_bound and element > result:
                result = element
        except TypeError:
            continue
    return result if result > 0 else None


def derivative_log_gamma_function(x):
    ascendence = np.arange(0, infinite_sum_upper_calculation_times)
    return -gamma_0 + np.sum(1 / (ascendence + 1) - 1 / (ascendence + x))


def self_iteration(x, alpha, beta, gamma):
    # n = len(x)
    # a = -math.pi**2 / 6 * n / gamma**2
    # b = (
    #     n * (math.log(beta) - 1 / gamma * derivative_log_gamma_function(alpha / gamma))
    #     + np.sum(np.log(x))
    #     - n / alpha
    #     + math.pi**2 / 6 * n * alpha / gamma**2
    # )
    # c = n
    # alpha = get_positive_largest(solve_quadratic_equation(a, b, c))
    # beta = (n * alpha / (gamma * np.sum(x**gamma))) ** (1 / gamma)
    # sign = np.sign(x - 1 / beta)
    # I = np.where(sign == -1, 0, sign)
    # d_4 = (
    #     -n
    #     * (1 / gamma + alpha / gamma**2 * derivative_log_gamma_function(alpha / gamma))
    #     + beta**gamma * np.dot(x**gamma, np.log((beta * x).astype("float")))
    #     + n
    #     * (
    #         2 / gamma
    #         - alpha * gamma_0 / gamma**2
    #         + (2 / 3 + math.pi**2 / 18) * alpha**2 / gamma**3
    #     )
    #     - np.sum(
    #         4
    #         * np.exp(
    #             (
    #                 2 * np.maximum(gamma * np.log((beta * x).astype("float")) - 1, 0)
    #             ).astype("float")
    #         )
    #         / gamma
    #         * I
    #         - 2 / (3 * gamma) * (1 - I)
    #     )
    # )
    # d_3 = (
    #     2 * gamma * (-d_4)
    #     - 2 * n
    #     - 2 / 3 * np.sum(1 - I)
    #     - np.sum(
    #         4
    #         * np.exp(
    #             (
    #                 2 * np.maximum(gamma * np.log((beta * x).astype("float")) - 1, 0)
    #             ).astype("float")
    #         )
    #         * I
    #     )
    # )
    # d_2 = 2 * gamma * (2 * n + 2 / 3 * np.sum(1 - I)) + gamma_0 * n * alpha
    # d_1 = -2 * gamma_0 * n * alpha * gamma - (2 / 3 + math.pi**2 / 18) * n * alpha**2
    # d_0 = (4 / 3 + math.pi**2 / 9) * n * gamma * alpha**2
    # gamma = get_positive_largest(
    #     solve_quartic_equation(d_4, d_3, d_2, d_1, d_0), 2 * gamma
    # )
    # return alpha, beta, gamma
    if (not hasattr(self_iteration, "x_id")) or id(x) != self_iteration.x_id:
        self_iteration.x_id = id(x)
        self_iteration.n = len(x)
        sign = np.sign(x - 1 / beta)
        self_iteration.I = np.where(sign == -1, 0, sign)
    a = -math.pi**2 / 6 * self_iteration.n / gamma**2
    b = (
        self_iteration.n
        * (math.log(beta) - 1 / gamma * derivative_log_gamma_function(alpha / gamma))
        + np.sum(np.log(x))
        - self_iteration.n / alpha
        + math.pi**2 / 6 * self_iteration.n * alpha / gamma**2
    )
    c = self_iteration.n
    alpha = get_positive_largest(solve_quadratic_equation(a, b, c))
    beta = (self_iteration.n * alpha / (gamma * np.sum(x**gamma))) ** (1 / gamma)
    tmp = 4 * np.exp(
        (2 * np.maximum(gamma * np.log((beta * x).astype("float")) - 1, 0)).astype(
            "float"
        )
    )
    d_4 = (
        -self_iteration.n
        * (1 / gamma + alpha / gamma**2 * derivative_log_gamma_function(alpha / gamma))
        + beta**gamma * np.dot(x**gamma, np.log((beta * x).astype("float")))
        + self_iteration.n
        * (
            2 / gamma
            - alpha * gamma_0 / gamma**2
            + (2 / 3 + math.pi**2 / 18) * alpha**2 / gamma**3
        )
        - np.sum(
            tmp / gamma * self_iteration.I - 2 / (3 * gamma) * (1 - self_iteration.I)
        )
    )
    d_3 = (
        2 * gamma * (-d_4)
        - 2 * self_iteration.n
        - 2 / 3 * np.sum(1 - self_iteration.I)
        - np.sum(tmp * self_iteration.I)
    )
    d_2 = (
        2 * gamma * (2 * self_iteration.n + 2 / 3 * np.sum(1 - self_iteration.I))
        + gamma_0 * self_iteration.n * alpha
    )
    d_1 = (
        -2 * gamma_0 * self_iteration.n * alpha * gamma
        - (2 / 3 + math.pi**2 / 18) * self_iteration.n * alpha**2
    )
    d_0 = (4 / 3 + math.pi**2 / 9) * self_iteration.n * gamma * alpha**2
    gamma = get_positive_largest(
        solve_quartic_equation(d_4, d_3, d_2, d_1, d_0), 2 * gamma
    )
    return alpha, beta, gamma


def self_simulation(x, alpha_0, beta_0, gamma_0, iteration_times=math.inf):
    alpha, beta, gamma = alpha_0, beta_0, gamma_0
    if math.inf != iteration_times:
        alpha_list, beta_list, gamma_list = [alpha], [beta], [gamma]
    i = 0
    while True:
        i += 1
        print("SeLF Iterations:", i)
        alpha, beta, gamma = self_iteration(x, alpha, beta, gamma)
        print(alpha, beta, gamma)
        if math.inf != iteration_times:
            alpha_list.append(alpha)
            beta_list.append(beta)
            gamma_list.append(gamma)
            if iteration_times == i:
                break
    return alpha_list, beta_list, gamma_list


def newton_iteration(x, alpha, beta, gamma):
    n = len(x)
    alpha -= (
        n * (math.log(beta) - 1 / gamma * derivative_log_gamma_function(alpha / gamma))
        + np.sum(np.log(x))
    ) / (
        -n
        / gamma**2
        * np.sum(
            1
            / (np.arange(0, infinite_sum_upper_calculation_times) + alpha / gamma) ** 2
        )
    )
    beta = (n * alpha / (gamma * np.sum(x**gamma))) ** (1 / gamma)
    gamma -= (
        n
        * (1 / gamma + alpha / gamma**2 * derivative_log_gamma_function(alpha / gamma))
        - beta**gamma * np.dot(x**gamma, np.log((beta * x).astype("float")))
    ) / (
        n
        * (
            -1 / gamma**2
            - 2 * alpha / gamma**3 * derivative_log_gamma_function(alpha / gamma)
            - alpha**2
            / gamma**4
            * np.sum(
                1
                / (np.arange(0, infinite_sum_upper_calculation_times) + alpha / gamma)
                ** 2
            )
        )
        - beta**gamma * np.dot(x**gamma, np.log((beta * x).astype("float")))
    )
    return alpha, beta, gamma


def newton_simulation(x, alpha_0, beta_0, gamma_0, iteration_times=math.inf):
    alpha, beta, gamma = alpha_0, beta_0, gamma_0
    if math.inf != iteration_times:
        alpha_list, beta_list, gamma_list = [alpha], [beta], [gamma]
    i = 0
    while True:
        i += 1
        print("Newton Iterations:", i)
        alpha, beta, gamma = newton_iteration(x, alpha, beta, gamma)
        print(alpha, beta, gamma)
        if math.inf != iteration_times:
            alpha_list.append(alpha)
            beta_list.append(beta)
            gamma_list.append(gamma)
            if iteration_times == i:
                break
    return alpha_list, beta_list, gamma_list


def output_simulation(
    n, alpha, beta, gamma, alpha_list, beta_list, gamma_list, iteration_times, tag
):
    with open(
        str(n)
        + "["
        + str(alpha)
        + ","
        + str(beta)
        + ","
        + str(gamma)
        + "]"
        + str(alpha_list[0])
        + ","
        + str(beta_list[0])
        + ","
        + str(gamma_list[0])
        + "["
        + str(iteration_times)
        + "]"
        + str(tag)
        + ".csv",
        mode="w+",
        encoding="utf-8",
    ) as output:
        output.write(
            str(alpha_list[0]) + "," + str(beta_list[0]) + "," + str(gamma_list[0])
        )
        for i in range(1, len(alpha_list)):
            output.write(
                "\n"
                + str(alpha_list[i])
                + ","
                + str(beta_list[i])
                + ","
                + str(gamma_list[i])
            )


def simulation(
    n, alpha, beta, gamma, alpha_0, beta_0, gamma_0, iteration_times=math.inf
):
    x = generalized_gamma(alpha, beta, gamma, n)
    alpha_list, beta_list, gamma_list = self_simulation(
        x, alpha_0, beta_0, gamma_0, iteration_times
    )
    output_simulation(
        n,
        alpha,
        beta,
        gamma,
        alpha_list,
        beta_list,
        gamma_list,
        iteration_times,
        "self",
    )
    alpha_list, beta_list, gamma_list = newton_simulation(
        x, alpha_0, beta_0, gamma_0, iteration_times
    )
    output_simulation(
        n,
        alpha,
        beta,
        gamma,
        alpha_list,
        beta_list,
        gamma_list,
        iteration_times,
        "newton",
    )


np.random.seed(1027)
simulation(1000, 2, 3, 2, 3, 2, 3, 200)

\end{lstlisting}

\section*{Acknowledgement}

I would like to express my sincere gratitude to Professor Guoliang Tian and his PhD student Xunjian Li for their invaluable guidance and assistance throughout the course of this research. Without their support, I would not have been able to complete this thesis. In particular, their development of the SeLF algorithm and US algorithm has made it possible for me to investigate this new topic related to the maximum likelihood estimation of the parameters in the generalized Gamma distribution.

I would also like to thank my roommate Xujian Chen for his unwavering support and companionship during the difficult times of this research. His help in reviewing the complex mathematical derivations has been essential to avoiding many errors and has allowed me to complete this thesis successfully.

Finally, I would like to express my gratitude to all the faculty members of the Department of Statistics and Data Science, the Department of Mathematics, and the Department of Computer Science and Engineering at Southern University of Science and Technology who have contributed to my academic growth. This research has broadened my horizons and enhanced my knowledge and skills, and I am deeply grateful for this invaluable experience.

Thank you all for your support and encouragement.

\end{document}